\documentclass[12pt]{article}
\usepackage{amsmath,amssymb,amsthm,bbm,natbib,color,hyperref}
\usepackage[margin=1in]{geometry}
\usepackage{graphicx}
\usepackage{tikz,verbatim}
\usetikzlibrary{decorations.pathreplacing}

\newtheorem{theorem}{Theorem}
\newtheorem{lemma}[theorem]{Lemma}
\newtheorem{proposition}[theorem]{Proposition}

\newtheorem{example}[theorem]{Example}
\newtheorem{fact}[theorem]{Fact}

\newcommand{\R}{\mathbb{R}}

\newcommand{\EE}[1]{\mathbb{E}\left[{#1}\right]}
\newcommand{\EEst}[2]{\mathbb{E}\left[{#1}\  \middle| \ {#2}\right]}
\newcommand{\Ep}[2]{\mathbb{E}_{{#1}}\left[{#2}\right]}

\newcommand{\PP}[1]{\mathbb{P}\left\{{#1}\right\}}
\newcommand{\PPst}[2]{\mathbb{P}\left\{{#1}\  \middle| \ {#2}\right\}}
\newcommand{\Ppst}[3]{\mathbb{P}_{{#1}}\left\{{#2}\  \middle| \ {#3}\right\}}
\newcommand{\Pp}[2]{\mathbb{P}_{{#1}}\left\{{#2}\right\}}
\newcommand{\eqd}{\stackrel{\textnormal{d}}{=}}
\newcommand{\One}[1]{{\mathbbm{1}}\left\{{#1}\right\}}
\newcommand{\one}[1]{{\mathbbm{1}}_{{#1}}}
\newcommand{\iidsim}{\stackrel{\textnormal{iid}}{\sim}}

  \newcommand\independent{\protect\mathpalette{\protect\independenT}{\perp}}
\def\independenT#1#2{\mathrel{\rlap{$#1#2$}\mkern2mu{#1#2}}}

\newcommand{\cX}{\mathcal{X}}
\newcommand{\cY}{\mathcal{Y}}
\newcommand{\cS}{\mathcal{S}}
\newcommand{\dtv}{\mathrm{d}_{\mathrm{TV}}}
\newcommand{\stleq}{\preceq_{\mathrm{st}}}
\newcommand{\stgeq}{\succeq_{\mathrm{st}}}
\newcommand{\cvxleq}{\preceq_{\mathrm{cvx}}}
\newcommand{\icxleq}{\preceq_{\mathrm{icx}}}
\newcommand{\dcxleq}{\preceq_{\mathrm{dcx}}}
\newcommand{\Fcal}{\mathcal{F}}
\newcommand{\Fleq}{\preceq_{\Fcal}}

\title{
Monte Carlo testing:
non-asymptotic guarantees without joint exchangeability
}
\author{Rina Foygel Barber\thanks{University of Chicago, \texttt{rina@uchicago.edu}} \ and Aaditya Ramdas\thanks{Carnegie Mellon University, \texttt{aramdas@cmu.edu}}}
\begin{document}
\maketitle

\begin{abstract}
In hypothesis testing, Monte Carlo tests are usually justified either by exact null simulation or by
joint exchangeability of the observed data and its simulated copies. This leaves
a gap for common computational procedures, such as parallel MCMC sampling
initialized at the observed data, where each copy may be marginally null and
even pairwise exchangeable with the observation, but the full collection is not
jointly exchangeable. In such cases the usual empirical p-value can be invalid
when the chain has not mixed, while exactly exchangeable constructions such as
the Besag--Clifford hub-and-spoke sampler may suffer from high conditional Monte
Carlo variability.
We give finite-sample guarantees for this intermediate regime. If, under the null, the observed
data \(X\) and a copy \(X'\sim P(\cdot\mid X)\) are conditionally
i.i.d.\ given a latent variable, then for any prespecified statistic and any finite
number $m$ of conditionally independent Monte Carlo copies, the resulting
empirical p-value obeys
\(
    \mathbb P\{p_m\le \alpha\}\le 2\alpha .
\)
This guarantee requires no mixing conditions and holds for any number of copies $m$, and it explains finite-sample oscillatory behavior in inference via MCMC sampling.  
In addition, we further show that the guarantee provides insights into inference problems arising in other settings, including inference on Bayesian models (recovering a classical result showing
validity up to a factor of $2$ for posterior predictive p-values), and inference via balanced permutation tests.
\end{abstract}

\section{Introduction}
Let $X\in\cX$ be observed data, and let $T:\cX\to\R$ be a prespecified test statistic. Suppose we wish to test a null hypothesis $H_0$, with the convention that large values of $T(X)$ indicate evidence against $H_0$. 

Of course, if the null distribution of $X$ is known, this immediately allows us to compute a p-value for testing the null, by computing the null distribution of $T(X)$. If instead we are only able to sample from the null, if $X_1,\dots,X_m$ are i.i.d.\ draws from the null distribution of $X$ then
\begin{equation}\label{eqn:pvalue_m_copies}p_m = \frac{1 + \sum_{i=1}^m\One{T(X_i)\geq T(X)}}{m+1}\end{equation}
provides a valid p-value for testing the null (often referred to as an `empirical p-value' or a `Monte Carlo p-value').

In many settings, however, computing (or sampling from) the null distribution of $X$ is not feasible: this distribution might not be known exactly, or the problem of sampling from the distribution may be computationally intractable. In these types of settings, alternative strategies may be used, as we describe next.

\subsection{Using MCMC sampling for inference}\label{sec:intro_MCMC}
Let $\pi(\cdot \mid x)$ be a Markov kernel, such that the null distribution $Q_0$ is a stationary distribution of the resulting Markov chain. We will write $\pi^s(\cdot\mid x)$ as the kernel for taking $s$ steps in the chain, for any $s\geq 1$.

If the Markov chain has good mixing properties, we may use MCMC samples as a proxy for i.i.d.\ samples: that is, we draw copies
\begin{equation}\label{eqn:MCMC_naive}X_1,\dots,X_m\iidsim \pi^s(\cdot \mid X)\end{equation}
by running the Markov chain for $s$ steps (for some sufficiently large $s$), initialized at $X$.

The (approximate) validity of this approach relies on the assumption that the Markov chain mixes well within $s$ steps. If this is not the case, there may be substantial dependence between $X$ and each copy $X_i$, and
this dependence may lead to an invalid p-value $p_m$, with an inflated Type-I error rate, $\Pp{H_0}{p_m\leq \alpha}>\alpha$. This can occur even if the Markov chain is reversible (i.e., even if each $X_i$ is exchangeable with $X$).

\begin{figure}[t]\centering
\includegraphics[width=0.75\textwidth]{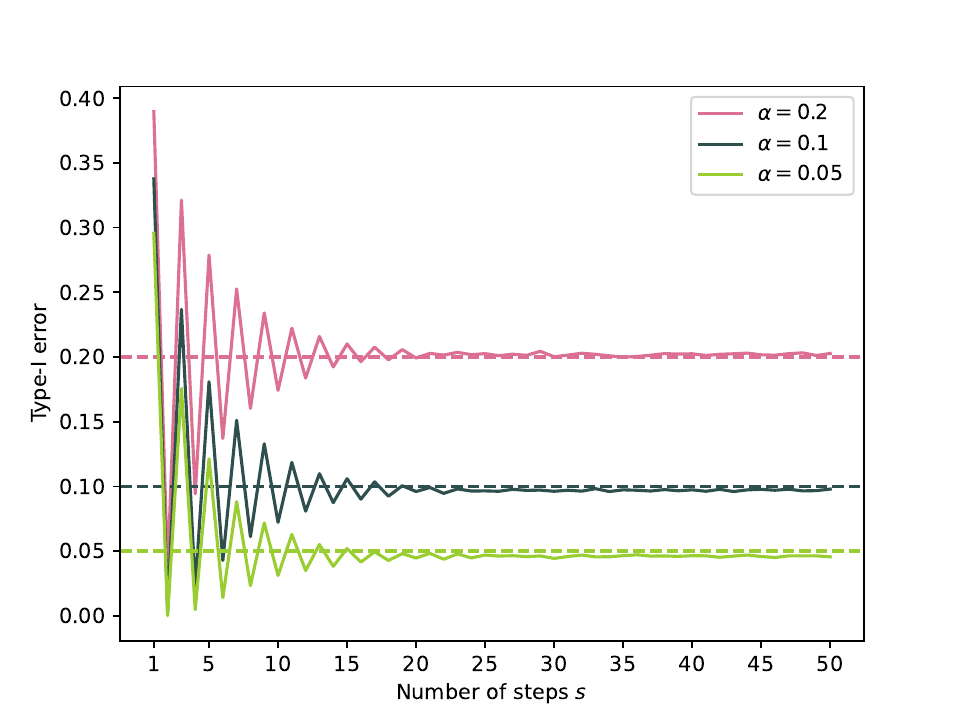}
    \caption{An example of Type-I error under the naive sampling scheme~\eqref{eqn:MCMC_naive}. The figure displays the Type-I error level $\PP{p_m\leq \alpha}$ for $m=1000$ (estimated over 10000 independent trials), as compared to the nominal level $\alpha$ (indicated by a dashed line).
    See Section~\ref{sec:intro_MCMC} for details. Instead of the traditional focus on asymptotic validity, our focus is on the nonasymptotic oscillations, and the surprising difference between even and odd values of $s$.
    }\label{fig:MCMC_naive_failure}
\end{figure}

See Figure~\ref{fig:MCMC_naive_failure} for an example, where we sample $X\sim \mathcal{N}(0,1)$, and the Markov kernel is $\pi(\cdot\mid x) = \mathcal{N}(\rho x, 1-\rho^2)$ for $\rho = -0.8$, with test statistic $T(X)=X$.\footnote{Code to reproduce empirical results is available at \url{https://colab.research.google.com/drive/1ze2I0keKNxwpe3E5mE2PLYTH-VJmQe04?usp=sharing}.} We can see that for moderate values of $s$, the quantity $p_m$ computed via the naive sampling strategy~\eqref{eqn:MCMC_naive} can massively fail to control Type-I error. As $s\to\infty$, on the other hand, $p_m$ behaves like a valid p-value, since the copies are now essentially i.i.d.\ draws from the distribution of $X$ (i.e., $s$ is sufficiently large for the Markov chain to exhibit mixing). Interestingly, we also observe that there is oscillation in the plots, and Type-I error does not decrease monotonically with $s$. Perhaps surprisingly, we will see later on that the worst-case Type-I error of this approach depends on whether $s$ is even or odd: an odd $s$ can potentially lead to arbitrarily large Type-I errors if the Markov chain is poorly mixing, while for even $s$ we will obtain a ``factor-of-$2$'' guarantee for any reversible Markov chain.

\subsubsection{Hub-and-spoke sampling}\label{sec:intro_hubspoke}
To overcome the potential loss of Type-I error control we have seen above, \citet{besag1989generalized} proposed a solution that modifies the way in which the copies are sampled from the Markov chain, and restore validity of $p_m$. Given a Markov kernel $\pi$ for which the null distribution $Q_0$ is stationary, their work proposes the following strategy (sometimes called `hub-and-spoke' sampling): 
\begin{equation}\label{eqn:MCMC_hub_and_spoke}\begin{cases}
    \textnormal{Sample a latent `hub' $X_*\sim \pi^{-s}(\cdot \mid X)$;}\\
\textnormal{Then sample the copies as `spokes', $X_1,\dots,X_m\iidsim \pi^s(\cdot\mid X_*)$.}
\end{cases}\end{equation}
Here $\pi^{-1}(\cdot \mid x)$ denotes the Markov kernel for the reverse chain, and $\pi^{-s}$ denotes sampling by taking $s$ steps in the reverse chain.
Figure~\ref{fig:besag_clifford} illustrates the hub-and-spoke sampling approach~\eqref{eqn:MCMC_hub_and_spoke} in comparison to the naive approach~\eqref{eqn:MCMC_naive}.\footnote{An alternative approach for generating copies is to sample serially, rather than in parallel: draw $X_1\sim \pi^s(\cdot\mid X)$, then $X_2\sim \pi^s(\cdot \mid X_1)$, and so on. \citet{besag1989generalized} also offer an exchangeable version of the serial sampling strategy, but we do not study these methods here.}

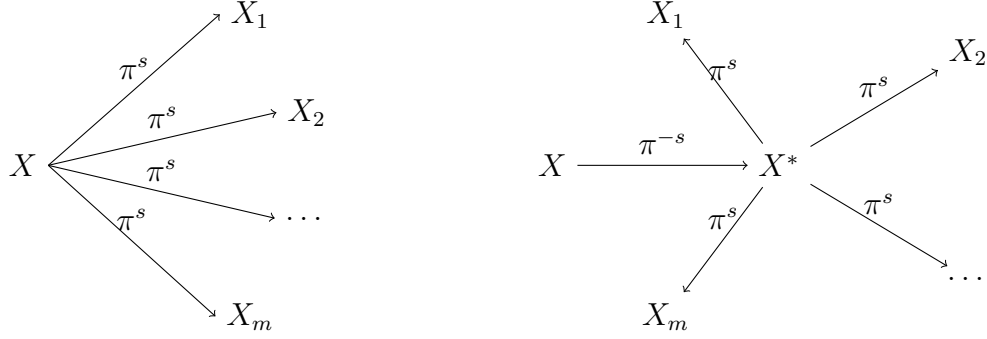
\begin{figure}[t]
\centering
\begin{tikzpicture}
\begin{scope}[xshift=0cm,yshift=0cm]

\node (X) at (-14,0) {$X$};
\node (Xt1) at (-11,2) {$X_1$};
\node (Xt2) at (-10.25,0.7) {$X_2$};
\node (Xt3) at (-10.25,-0.7) {$\dots$};
\node (Xt4) at (-11,-2) {$X_m$};
\draw[->] (X.east) -- (Xt1.west) node[midway,above] () {$\pi^s$};
\draw[->] (X.east) -- (Xt2.west) node[midway,above] () {$\pi^s$};
\draw[->] (X.east) -- (Xt3.west) node[midway,above] () {$\pi^s$};
\draw[->] (X.east) -- (Xt4.west) node[midway,above] () {$\pi^s$};

\node (X_) at (-7,0) {$X$};
\node (Xhub_) at (-4,0) {$X^*$};
\node (Xt1_) at (-5.5,2) {$X_1$};
\node (Xt2_) at (-1.5,1.5) {$X_2$};
\node (Xt3_) at (-1.5,-1.5) {$\dots$};
\node (Xt4_) at (-5.5,-2) {$X_m$};
\draw[->] (X_) -- (Xhub_) node[midway,above] () {$\pi^{-s}$};
\draw[->] (Xhub_) -- (Xt1_) node[midway,above] () {$\pi^s$};
\draw[->] (Xhub_) -- (Xt2_) node[midway,above] () {$\pi^s$};
\draw[->] (Xhub_) -- (Xt3_) node[midway,above] () {$\pi^s$};
\draw[->] (Xhub_) -- (Xt4_) node[midway,above] () {$\pi^s$};

\end{scope}
\end{tikzpicture}
    \caption{An illustration of the naive sampling strategy~\eqref{eqn:MCMC_naive} (left) and \citet{besag1989generalized}'s hub-and-spoke sampling strategy~\eqref{eqn:MCMC_hub_and_spoke} (right).}
    \label{fig:besag_clifford}
\end{figure}

The hub-and-spoke sampling scheme satisfies the following property:
\begin{equation}\label{eqn:joint_exch}\textnormal{If $X\sim Q_0$, and $Q_0$ is stationary under $\pi$, then $(X,X_1,\dots,X_m)$ is exchangeable.}\end{equation}
This property directly implies validity of the p-value $p_m$, regardless of the mixing properties of the Markov chain: $p_m$ is a valid p-value for testing $H_0$, even under arbitrarily strong dependence. However, poor mixing can lead to a different issue: that of excessive randomness, where $\textnormal{Var}(p_m\mid X)$ is nonnegligible even for arbitrarily large $m$. This is because $p_m$ may depend strongly on the randomly drawn hub $X_*$.
See Figure~\ref{fig:besag_clifford_highvar} for an illustration (in the same setting as Figure~\ref{fig:MCMC_naive_failure}, with $s=5$).
\begin{figure}[t]
\centering
\includegraphics[width=0.8\textwidth]{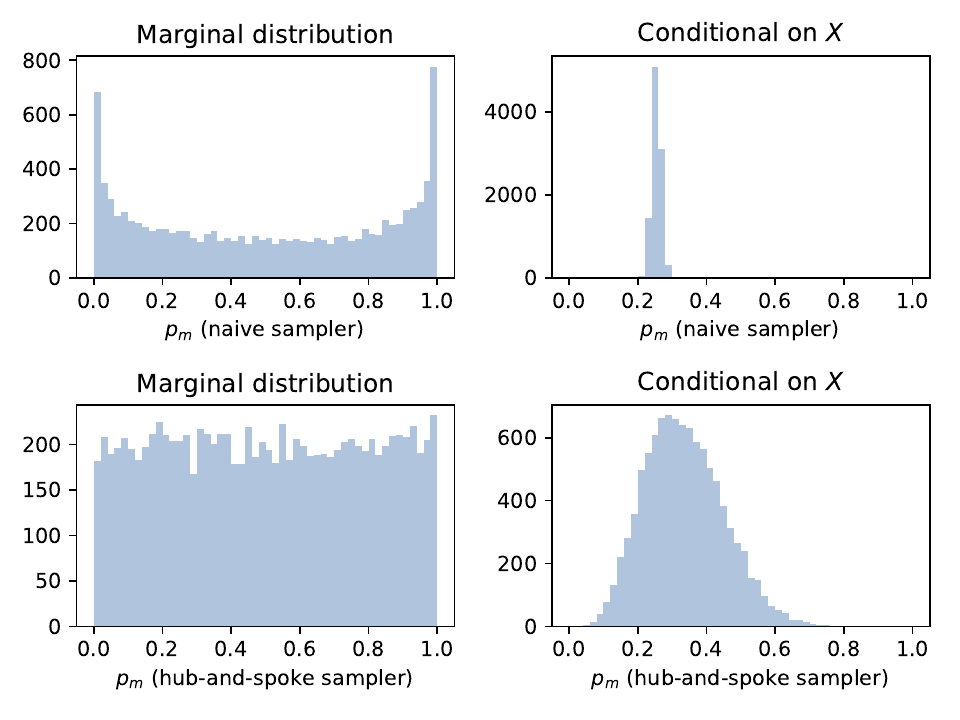}
    \caption{The left plots show the marginal distribution of $p_m$ ($m=1000,s=5$), and the right plots show conditional distribution of $p_m$ given a single draw of the data $X$, for naive~\eqref{eqn:MCMC_naive} or hub-and-spoke~\eqref{eqn:MCMC_hub_and_spoke} sampling. We see that hub-and-spoke sampling offers a valid p-value (i.e., the bottom-left histogram is uniform), but this comes at the cost of substantial randomness (i.e., the bottom-right histogram shows high variability). See Section~\ref{sec:intro_hubspoke} for details.}
    \label{fig:besag_clifford_highvar}
\end{figure}

In contrast, the sampling strategy given in~\eqref{eqn:MCMC_naive} cannot have this issue: we must have $\textnormal{Var}(p_m\mid X)\lesssim 1/m$, since $p_m$ averages over $m$ i.i.d.\ draws (conditional on $X$). However, as we have already seen, this can come at a cost: loss of Type-I error control.

\subsection{Our contributions}
In this work, we will establish that the MCMC approach defined in~\eqref{eqn:MCMC_naive}, which in general does not yield copies satisfying joint exchangeability~\eqref{eqn:joint_exch}, can nonetheless offer a (weaker) Type-I error guarantee, in certain settings. In particular, this means that we can avoid the issue of excessive randomness that can arise in \citet{besag1989generalized}'s approach, while maintaining a Type-I error guarantee that does not rely on any mixing conditions.

\section{Theoretical guarantees}
Given the observed data $X$, we will study the p-value
\begin{equation}\label{eqn:sample_copies_P}p_m =\frac{1 + \sum_{i=1}^m\One{T(X_i)\geq T(X)}}{m+1},  \textnormal{ where }X_1,\dots,X_m\mid X \iidsim P(\cdot\mid X),\end{equation}
for some choice of probability kernel $P$, and some prespecified test statistic $T:\cX\to\R$. In order for this to be a reasonable approach for testing the null hypothesis, we require that $P$ is \emph{compatible with the null}, in the following sense:
\begin{equation}\label{eqn:P_compatible_with_null}
    \textnormal{If $X\sim Q_0$ and $X'\mid X\sim P(\cdot\mid X)$, then marginally $X'\sim Q_0$.}
\end{equation}
For example, in the setting of naive MCMC sampling in~\eqref{eqn:MCMC_naive}, we choose $P(\cdot\mid X) = \pi^s(\cdot\mid X)$, where $\pi$ is the Markov kernel and $s$ is the number of steps. In this case, the condition~\eqref{eqn:P_compatible_with_null} is satisfied as long as the null distribution $Q_0$ is a stationary distribution for the Markov chain.

\subsection{A guarantee under the forward--backward condition}\label{sec:thm1}
Our first main result considers probability kernels $P$ that satisfy an additional condition:\footnote{Throughout the paper, we will implicitly assume that $\cX$ is a standard Borel space, to provide standard regularity conditions for working with conditional distributions and conditional expectations \citep[Theorem 10.2.2]{dudley2018real}.}
\begin{equation}\label{eqn:fb_condition}\begin{split}\textnormal{If $X\sim Q_0$ and $X'\mid X\sim P(\cdot\mid X)$, then there exists a random }\\
\textnormal{variable $Y$ such that $X,X'$ are conditionally i.i.d.\ given $Y$.}\end{split}
\end{equation}
Note that this condition is defined relative to a null distribution $Q_0$. In particular, it implies compatibility with the null~\eqref{eqn:P_compatible_with_null}, and moreover implies that $(X,X_i)$ will be an exchangeable pair for each copy $i=1,\dots,m$ (but does not ensure joint exchangeability as in~\eqref{eqn:joint_exch}).

We will refer to any $P$ satisfying the condition~\eqref{eqn:fb_condition} as a \emph{forward--backward} probability kernel, for the following reason:  if the condition~\eqref{eqn:fb_condition} is satisfied, then writing $\tilde{P}$ to denote the joint distribution of $(X,Y)$, sampling $X'\mid X\sim P(\cdot \mid X)$ is equivalent to first drawing $Y$ from the conditional distribution $\tilde{P}_{Y\mid X}(\cdot \mid X)$, and then drawing $X'$ from the conditional distribution $\tilde{P}_{X\mid Y}(\cdot \mid Y)$. In other words, if we consider a Markov chain on $\cX \cup \cY$, with Markov kernel $\tilde{P}_{Y\mid X}$, then sampling $X'\mid X\sim P(\cdot \mid X)$ is equivalent to taking one forward step on this chain (to draw $Y$ given $X$) followed by one backward step on this chain (to draw $X'$ given $Y$).

We are now ready to present our first bound on Type-I error.
\begin{theorem}\label{thm1}
    Let $P$ be a forward--backward probability kernel relative to the null $Q_0$. Then, if $X\sim Q_0$, the p-value $p_m$ defined in~\eqref{eqn:sample_copies_P} satisfies
    \[\PP{p_m\leq \alpha}\leq 2\alpha\textnormal{ for all $\alpha\in[0,1]$.}\]
\end{theorem}
\noindent This result holds for any finite $m$, and we do not require $m$ to be large---but of course, a larger $m$ will often lead to better performance, in terms of allowing for smaller values of $p_m$ (since $p_m$ can never be smaller than $\frac{1}{m+1}$), and also reducing randomness, i.e., reducing $\textnormal{Var}(p_m\mid X)$. It is also of interest, therefore, to consider the limiting case, as $m\to\infty$. Define
\begin{equation}\label{eqn:sample_copies_P_infty}p_\infty = \PPst{T(X')\geq T(X)}{X},\textnormal{ where }X'\mid X \sim P(\cdot\mid X).\end{equation}
Conditional on $X$, the finite-$m$ quantity $p_m$ is simply an empirical estimate of $p_\infty$, and we must have $p_m\stackrel{\textnormal{a.s.}}{\to} p_\infty$ as $m\to\infty$; we can think of $p_\infty$ as a completely derandomized version of the p-value $p_m$. 
Consequently, $p_\infty$ inherits the same Type-I error guarantees as $p_m$: that is, Theorem~\ref{thm1} holds with $m=\infty$ as well.

In fact, we will also see below that the case $m=\infty$ can be established with existing tools, because $p_\infty$ can be represented as an `average of valid p-values' \citep{ruschendorf1982random,meng1994posterior,vovk2020combining,wang2024testing}.
In contrast, the result of Theorem~\ref{thm1} for finite $m$ requires a new type of argument.

\paragraph{An explanation of Figure~\ref{fig:MCMC_naive_failure}.}
    To better understand the implications of this theorem, we return to the oscillatory behavior observed for low values of $s$ in Figure~\ref{fig:MCMC_naive_failure}. We will now see why Theorem~\ref{thm1} explains this oscillation.
    
    Consider a Markov chain with transition probabilities given by the Markov kernel $\pi$, and with stationary distribution $Q_0$. Suppose the Markov chain is reversible. Fix any $s\geq 1$, and let $P(\cdot\mid X) = \pi^s(\cdot\mid X)$.     
    Then, if the number of steps $s$ is even, $P$ is a forward--backward probability kernel: if we write $s=2r$ and let $Y\sim \pi^r(\cdot\mid X)$ denote a random variable obtained by taking $r$ steps along the Markov chain, then the condition~\eqref{eqn:fb_condition} is clearly satisfied.

 This example explains the pattern observed in Figure~\ref{fig:MCMC_naive_failure}: for even values of $s$, the Type-I error is guaranteed to be bounded (in fact, in the example shown in the figure we have Type-I error $\leq \alpha$, although the theorem only guarantees $\leq 2\alpha$). In contrast, for odd values of $s$ it is no longer the case that $P$ is a forward--backward probability kernel, and the Type-I error may be quite high---even though it still holds that $P$ is compatible with the null as in~\eqref{eqn:P_compatible_with_null}, and even though  $(X,X_i)$ is an exchangeable pair for each copy $i=1,\dots,m$.

From this example, we draw the following conclusion. If we use a reversible MCMC sampler for generating copies, then it is safest to use an even step size $s$, as a safeguard against the possibility of slow mixing: while we hope that the Markov chain mixes well within $s$ steps, so that the copies are nearly i.i.d., even if this is not the case we would lose at most a factor of $2$ in the Type-I error control.

\subsection{An alternative bound under a total variation condition}
In the discussion above, we applied the results of Theorem~\ref{thm1} to the setting of MCMC sampling, where copies $X_i$ are generated by taking $s$ steps (for an even $s$) along a reversible Markov chain. So far, our results have not placed any assumptions on the mixing properties of this chain: the Type-I bound of Theorem~\ref{thm1} applies even for small $s$, and even if the Markov chain is very slowly mixing (so that the copies $X_i$ may be highly correlated with $X$). However, empirically in Figure~\ref{fig:MCMC_naive_failure} we observe that as $s$ increases, the p-value $p_m$ becomes more reliable, i.e., its Type-I error approaches the nominal level $\alpha$---and moreover, this limiting behavior holds for both odd and even $s$. This is because, as $s\to\infty$, the dependence between $X$ and its copies is vanishing, and so the copies are essentially i.i.d.\ draws from the same distribution $Q_0$ as $X$. Our next result examines why this occurs.

Returning to the general setting where copies are generated from any probability kernel $P$, as in~\eqref{eqn:sample_copies_P}, we now aim to show that the resulting p-value is approximately valid if there is limited dependence between $X$ and its copies $X_i$, that is, if the $X_i$'s are nearly i.i.d.\ copies of $X$:
\begin{equation}\label{eqn:P_tv_condition_simple}
    \textnormal{If $X\sim Q_0$ and $X'\mid X\sim P(\cdot\mid X)$, then $\dtv\big((X,X'), Q_0\times Q_0\big)\leq \epsilon$},
\end{equation}
where $\dtv$ denotes the total variation distance.
This condition can be generalized to the following: there exists some function $f$ such that
\begin{equation}\label{eqn:P_tv_condition}\begin{split}
    \textnormal{If $X\sim Q_0$ and $X'\mid X\sim P(\cdot\mid X)$ and $X''\mid X \sim Q_0(\cdot\mid f(X))$,}\\\textnormal{then $\dtv\big((X,X'), (X,X'')\big)\leq \epsilon$.}\end{split}
\end{equation}
Here $Q_0(\cdot\mid f(X))$ denotes the conditional distribution of $X\mid f(X)$ induced by $X\sim Q_0$.
(The simpler condition~\eqref{eqn:P_tv_condition_simple} can be obtained as a special case by simply taking $f(x)\equiv 0$, i.e., $f(X)$ contains no information.)

We may also consider stronger total variation bounds:
\begin{equation}\label{eqn:P_tv_condition_simple_strong}
    \textnormal{If $X\sim Q_0$, then $\dtv\big(P(\cdot\mid X),Q_0\big)\leq \epsilon$ almost surely},
\end{equation}
or more generally, for some function $f$,
\begin{equation}\label{eqn:P_tv_condition_strong}
    \textnormal{If $X\sim Q_0$ then $\dtv\big(P(\cdot\mid X),Q_0(\cdot\mid f(X))\big)\leq \epsilon$ almost surely.}
\end{equation}

\begin{theorem}\label{thm2}
    Let $P$ be a probability kernel that 
    satisfies the condition~\eqref{eqn:P_tv_condition}.
    Let the p-value $p_m$ be defined in~\eqref{eqn:sample_copies_P} for finite $m\geq1$, or, as defined in~\eqref{eqn:sample_copies_P_infty} for $m=\infty$. Then, if $X\sim Q_0$, $p_m$ satisfies
    \[\PP{p_m\leq \alpha}\leq \alpha+\sqrt{2\epsilon}\textnormal{ for all $\alpha\in[0,1]$.}\]
    If instead the stronger condition~\eqref{eqn:P_tv_condition_strong} is satisfied,
    then
    \[\PP{p_m\leq \alpha}\leq \alpha+\epsilon\textnormal{ for all $\alpha\in[0,1]$.}\]
\end{theorem}
\noindent Note that in this setting, there is no longer a multiplicative factor of $2$ in the Type-I error guarantee, unlike in Theorem~\ref{thm1} where the bound is $2\alpha$. Instead, the Type-I error may be arbitrarily close to $\alpha$, depending on the parameter $\epsilon$ that describes the probability kernel $P$ as in~\eqref{eqn:P_tv_condition} or~\eqref{eqn:P_tv_condition_strong}.

\paragraph{An explanation of Figure~\ref{fig:MCMC_naive_failure}, revisited.}

We now see how this relates to the setting where $P(\cdot\mid X) = \pi^s(\cdot\mid X)$ for a Markov chain.

    Consider a Markov chain with transition probabilities given by the Markov kernel $\pi$, and with stationary distribution $Q_0$. Fix any $s\geq 1$, and let $P(\cdot\mid X) = \pi^s(\cdot\mid X)$. Suppose the Markov chain satisfies the following mixing property:
    \[\dtv\big( \pi^s(\cdot\mid X),Q_0\big)\leq \epsilon\textnormal{ almost surely}.\]
    In the terminology of MCMC, this means that the \emph{mixing time} of the Markov chain is $\leq s$ (for tolerance level $\epsilon$). Then the probability kernel $P$ satisfies the condition~\eqref{eqn:P_tv_condition_strong} (in fact, the simpler condition~\eqref{eqn:P_tv_condition_simple_strong} is satisfied), and so Theorem~\ref{thm2} implies $\PP{p_m\leq\alpha}\leq\alpha+\epsilon$.

This application explains the phenomenon we observe in Figure~\ref{fig:MCMC_naive_failure} as $s\to\infty$: for sufficiently large $s$ (regardless of whether $s$ is even or odd), once the Markov chain has mixed reasonably well, we see that $p_m$ provides an approximately valid p-value.

\section{Applications}
We next develop several applications of Theorem~\ref{thm1} to a range of statistical inference problems. We first return to the motivating question of inference with MCMC samples, now exploring the more general setting of a non-reversible Markov chain. Afterwards, we consider additional examples: Bayesian posterior predictive inference, and constrained permutation testing.

\subsection{Reversible and non-reversible MCMC}\label{sec:examples_nonreversible}

For our first application, we return to the problem of using MCMC sampling to generate the copies $X_i$, as described in Section~\ref{sec:intro_MCMC}. Here we will consider three different ways to apply the results of Theorem~\ref{thm1} to this problem.

First, to provide a baseline, we review the result described in Section~\ref{sec:thm1}, for the case of a reversible Markov chain.
\begin{example}[Forward--backward probability kernel for a reversible Markov chain]\label{ex1}
    Fix any $s\geq 1$, and define $P(\cdot\mid X) = \pi^s(\cdot\mid X)$, where $\pi$ is the transition kernel for a reversible Markov chain with stationary distribution $Q_0$.
    Then, if the number of steps $s=2r$ is even, $P$ is a forward--backward probability kernel.
\end{example}

We emphasize that Example~\ref{ex1} above applies only for Markov chains that are reversible (i.e., $\pi = \pi^{-1}$). If the transition kernel $\pi$ corresponds to a non-reversible Markov chain, then under slow mixing, the Type-I error control properties of $p_m$ may be arbitrarily bad regardless of whether $s$ is odd or even---see Figure~\ref{fig:MCMC_naive_failure_nonreversible} for an example, where we sample $X\sim Q_0$ for $Q_0=\textnormal{Unif}[0,1]$ with test statistic $T(X)=X$, and use the transition kernel $\pi(\cdot\mid x) = (1/\tau) \cdot \delta_{x/\tau} + (1-1/\tau)\cdot\textnormal{Unif}[1/\tau,1]$, with $\tau = 1.01$.

\begin{figure}[t]
\centering
\includegraphics[width=0.8\textwidth]{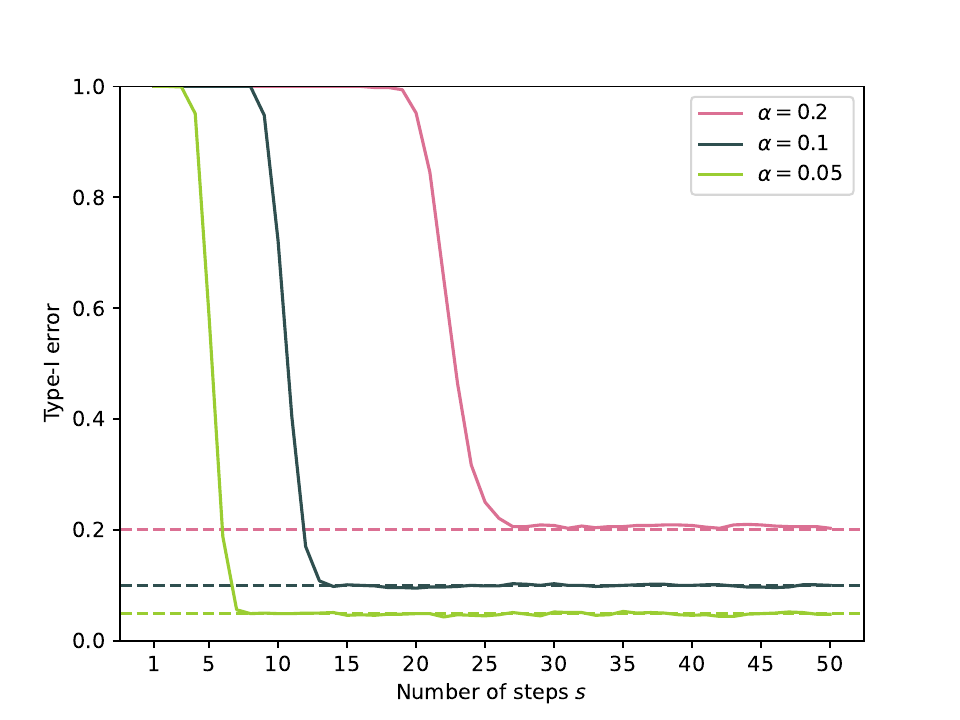}
    \caption{An example of Type-I error under the naive sampling scheme~\eqref{eqn:MCMC_naive}, when the Markov chain is not reversible. The figure displays the Type-I error level $\PP{p_m\leq \alpha}$ for $m=1000$ (estimated over 1000 independent trials), as compared to the nominal level $\alpha$ (indicated by a dashed line); we see severe loss of Type-I error control for low values of $s$, regardless of even or odd values (in contrast to Figure~\ref{fig:MCMC_naive_failure})
    See Section~\ref{sec:examples_nonreversible} for details.
    }\label{fig:MCMC_naive_failure_nonreversible}
\end{figure}

However, even in the case of a non-reversible Markov chain, we can modify the sampling strategy in order to be able to apply Theorem~\ref{thm1}. Here we present two such options  (where for each one, we again consider a Markov chain with transition probabilities $\pi$ for which $Q_0$ is stationary, but no longer assume it is reversible).

\begin{example}[Forward--backward probability kernel for a non-reversible Markov chain: version 1]\label{ex2}
     Fix any $r\geq 1$, and define the probability kernel $P$ as $P(\cdot\mid X) = [ \pi^r \circ \pi^{-r}](\cdot \mid X)$, where $\pi^r \circ \pi^{-r}$ denotes that we first take $r$ steps backward in the chain, then $r$ steps forward in the chain.
    Note that the conditional distribution of each individual copy, $X_i\mid X$, is the same as for \citet{besag1989generalized}'s hub-and-spoke sampler~\eqref{eqn:MCMC_hub_and_spoke} (with $r$ in place of $s$).
\end{example}

\begin{example}[Forward--backward probability kernel for a non-reversible Markov chain: version 2]\label{ex3}
    Fix any $r\geq 1$, and define the probability kernel $P$ as $P(\cdot\mid X) = [ \pi^{-1}\circ \pi]^r(\cdot \mid X)$. That is, we take one step forward and one step backward in the chain, and then repeat this process $r$ times.
\end{example}

 \noindent Note that, if the Markov chain is reversible, then Examples~\ref{ex2} and~\ref{ex3} are in fact exactly equivalent to Example~\ref{ex1}.
See Figure~\ref{fig:illustrate_MCMC_options} for an illustration comparing all three of these MCMC examples. The following proposition verifies the validity of each of the above examples:
\begin{proposition}\label{prop:MCMC_fb}
    In each MCMC example above (Examples~\ref{ex1},~\ref{ex2}, and~\ref{ex3}), the kernel $P$ is a forward--backward probability kernel.
\end{proposition}
\noindent Consequently, in each example, Theorem~\ref{thm1} ensures that constructing $p_m$ with copies sampled from $P$ will satisfy $\PP{p_m\leq\alpha}\leq 2\alpha$, for any finite $m$ or for $m=\infty$, under the null $X\sim Q_0$.
\begin{proof}[Proof of Proposition~\ref{prop:MCMC_fb}]
    For Example~\ref{ex1}, the condition~\eqref{eqn:fb_condition} is satisfied as explained in Section~\ref{sec:thm1}: we define $Y  \sim\pi^r(\cdot\mid X)$, taking $r$ steps forward in the chain from $X$. Then $X,X'\mid Y\iidsim \pi^r(\cdot\mid Y)$. 
    
    For Example~\ref{ex2}, the condition~\eqref{eqn:fb_condition} is satisfied by defining $Y\sim \pi^{-r}(\cdot\mid X)$, i.e., taking $r$ steps backward in the chain. Then we again have $X,X'\mid Y\iidsim \pi^r(\cdot\mid Y)$. 
    
    For Example~\ref{ex3}, the condition~\eqref{eqn:fb_condition} is satisfied by defining $Y\sim [ \pi^{-1}\circ \pi]^{r/2}(\cdot \mid X)$ (if $r$ is even), or $Y\sim \big[\pi\circ [ \pi^{-1}\circ \pi]^{(r-1)/2}\big](\cdot \mid X)$ (if $r$ is odd), and we then have $X,X'\mid Y\iidsim  [ \pi^{-1}\circ \pi]^{r/2}(\cdot \mid Y)$ (if $r$ is even) or $X,X'\mid Y\iidsim \big[[ \pi^{-1}\circ \pi]^{(r-1)/2}\circ\pi^{-1}\big](\cdot \mid Y)$ (if $r$ is odd).
\end{proof}

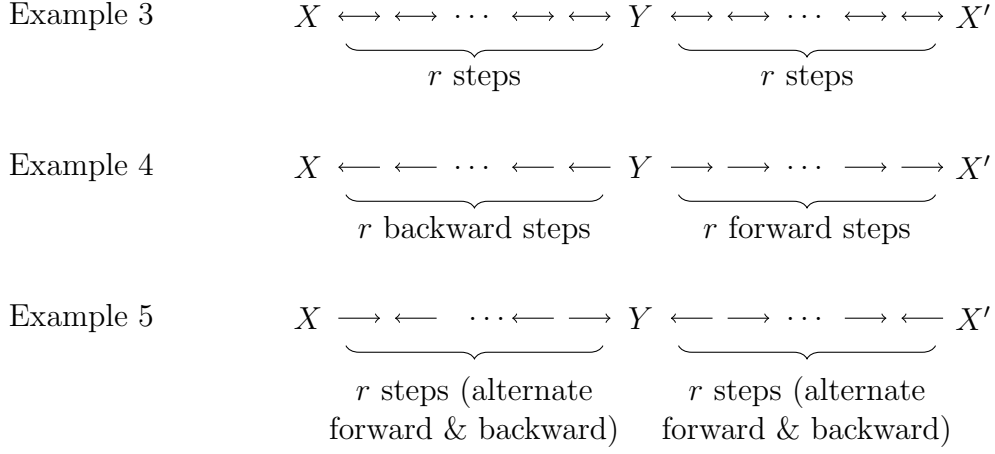
\begin{figure}[t]
\centering
\begin{tikzpicture}
\begin{scope}[xshift=0cm,yshift=0cm]

\node at (-8,2) {Example~\ref{ex1}};
\node (X) at (-5,2) {$X$};
\node (Y) at (-0.6,2) {$Y$};
\node (Xp) at (3.8,2) {$X'$};
\draw[<->] (-4.6,2) -- (-4.05,2);
\draw[<->] (-3.85,2) -- (-3.3,2);
\node at (-2.8,2) {$\dots$};
\draw[<->] (-2.3,2) -- (-1.75,2);
\draw[<->] (-1.55,2) -- (-1,2);
\draw[<->] (-0.2,2) -- (0.35,2);
\draw[<->] (0.55,2) -- (1.1,2);
\node at (1.6,2) {$\dots$};
\draw[<->] (2.1,2) -- (2.65,2);
\draw[<->] (2.85,2) -- (3.4,2);
\draw [decorate,decoration={brace,amplitude=5pt,mirror,raise=2ex}]
  (-4.5,2) -- (-1.1,2) node[midway,yshift=-2em]{$r$ steps};
\draw [decorate,decoration={brace,amplitude=5pt,mirror,raise=2ex}]
  (-0.1,2) -- (3.3,2) node[midway,yshift=-2em]{$r$ steps};

\node at (-8,0) {Example~\ref{ex2}};
\node (X) at (-5,0) {$X$};
\node (Y) at (-0.6,0) {$Y$};
\node (Xp) at (3.8,0) {$X'$};
\draw[<-] (-4.6,0) -- (-4.05,0);
\draw[<-] (-3.85,0) -- (-3.3,0);
\node at (-2.8,0) {$\dots$};
\draw[<-] (-2.3,0) -- (-1.75,0);
\draw[<-] (-1.55,0) -- (-1,0);
\draw[->] (-0.2,0) -- (0.35,0);
\draw[->] (0.55,0) -- (1.1,0);
\node at (1.6,0) {$\dots$};
\draw[->] (2.1,0) -- (2.65,0);
\draw[->] (2.85,0) -- (3.4,0);
\draw [decorate,decoration={brace,amplitude=5pt,mirror,raise=2ex}]
  (-4.5,0) -- (-1.1,0) node[midway,yshift=-2em]{$r$ backward steps};
\draw [decorate,decoration={brace,amplitude=5pt,mirror,raise=2ex}]
  (-0.1,0) -- (3.3,0) node[midway,yshift=-2em]{$r$ forward steps};

\node at (-8,-2) {Example~\ref{ex3}};
\node (X) at (-5,-2) {$X$};
\node (Y) at (-0.6,-2) {$Y$};
\node (Xp) at (3.8,-2) {$X'$};
\draw[->] (-4.6,-2) -- (-4.05,-2);
\draw[<-] (-3.85,-2) -- (-3.3,-2);
\node at (-2.6,-2) {$\dots$};
\draw[<-] (-2.3,-2) -- (-1.75,-2);
\draw[->] (-1.55,-2) -- (-1,-2);
\draw[<-] (-0.2,-2) -- (0.35,-2);
\draw[->] (0.55,-2) -- (1.1,-2);
\node at (1.6,-2) {$\dots$};
\draw[->] (2.1,-2) -- (2.65,-2);
\draw[<-] (2.85,-2) -- (3.4,-2);
\draw [decorate,decoration={brace,amplitude=5pt,mirror,raise=2ex}]
  (-4.5,-2) -- (-1.1,-2) node[midway,yshift=-3em]{\begin{tabular}{c}$r$ steps (alternate\\ forward \& backward)\end{tabular}};
\draw [decorate,decoration={brace,amplitude=5pt,mirror,raise=2ex}]
  (-0.1,-2) -- (3.3,-2) node[midway,yshift=-3em]{\begin{tabular}{c}$r$ steps (alternate\\ forward \& backward)\end{tabular}};

\end{scope}
\end{tikzpicture}
\caption{An illustration of the forward--backward probability kernels $P(\cdot\mid X)$ considered in Examples~\ref{ex1},~\ref{ex2}, and~\ref{ex3}. In the figure, a forward arrow $\longrightarrow$ represents a forward step in the Markov chain (with Markov kernel $\pi$), while the reverse arrow $\longleftarrow$ represents a backward step in the chain (i.e., sample from $\pi^{-1}$). In the top panel (Example~\ref{ex1}), the Markov chain is reversible, thus $\pi$ and $\pi^{-1}$ are equivalent and so each step is represented by a two-sided arrow $\longleftrightarrow$.}
\label{fig:illustrate_MCMC_options}
\end{figure}

\subsection{Bayesian posterior predictive inference}
While the examples so far have all considered variants of MCMC sampling, the general framework of drawing copies of $X$ in order to perform inference appears more broadly across different statistical settings, and our results may be applied in other contexts as well. Here we consider the problem of Bayesian posterior predictive inference.

Consider a Bayesian model consisting of a prior $\nu$ on $\theta\in\Theta$, and a  family of likelihoods $\{f_\theta:\theta\in\Theta\}$. Suppose we observe data $X\in\cX$ and would like to validate this Bayesian model. We can define the posterior predictive distribution (PPD) given $X$ as
    \[f_{\textnormal{PPD}}(x \mid X) = \Ep{\theta\sim \nu(\cdot\mid X)}{f_\theta(x)},\]
    where $\nu(\cdot\mid X)$ denotes the posterior distribution of $\theta\mid X$. We may then use the PPD for model validation, by sampling copies $X_1,\dots,X_m\mid X\sim f_{\textnormal{PPD}}(\cdot\mid X)$ and comparing these copies to $X$ (with some test statistic) \citep{meng1994posterior}.

We therefore have the following proposition:
\begin{proposition}\label{prop:PPD}
     Given data $X\in\cX$, let the copies $X_1,\dots,X_m$ be generated as
     \[X_1,\dots,X_m\mid X \iidsim f_{\mathrm{PPD}}(\cdot\mid X).\]
     Then, if $X$ is drawn from the Bayesian model specified by prior $\nu$ and likelihoods $\{f_\theta\}$, we have
     \[\PP{p_m\leq \alpha}\leq 2\alpha\textnormal{ for all }\alpha\in[0,1],\]
     for any (finite or infinite) $m$.
\end{proposition}
(This result was established by \citet{meng1994posterior} for the case $m=\infty$; the framework of Theorem~\ref{thm1} allows us to generalize to finite $m$ at no cost.)
\begin{proof}[Proof of Proposition~\ref{prop:PPD}]
       By construction, the probability kernel $f_{\textnormal{PPD}}(\cdot\mid X)$ is compatible with the null, as in~\eqref{eqn:P_compatible_with_null}.
    Moreover, $f_{\textnormal{PPD}}$
    is a forward--backward probability kernel relative to~\eqref{eqn:fb_condition}: we can verify this by observing that the condition~\eqref{eqn:fb_condition} is satisfied by choosing $Y=\theta$ (since sampling from $f_{\textnormal{PPD}}$ is equivalent to first sampling $\theta\mid X$, and then resampling $X\mid\theta$, under the Bayesian model). Therefore, the result holds by Theorem~\ref{thm1}.
\end{proof}

\subsection{Constrained permutation testing}

In permutation testing, Constrained permutation schemes arise when full permutation is mathematically
valid but scientifically or observationally undesirable, or when only certain
rearrangements are admissible. Examples include restricted-position
permutations for truncated data, local phenotype permutations in genetic
association studies with population stratification, spatially constrained
permutations in ecological association testing, and balanced permutations in
treatment-effect problems. These constrained sets are often not subgroups, so
the usual exact permutation argument can fail.

To describe the setting more precisely, we begin by reviewing a classical (full) permutation test. Given $n$ data points $Y_1,\dots,Y_n$ along with a test statistic $T:\cY^n\to \R$, the classical permutation test compares a test statistic of the data to all possible permutations,
\[p_\infty = \frac{\sum_{\sigma\in\cS_n}\One{T(Y_{\sigma(1)},\dots,Y_{\sigma(n)})\geq T(Y_1,\dots,Y_n)}}{n!},\]
or we may consider a Monte Carlo version by choosing $m$ permutations,
\[p_m = \frac{1+\sum_{i=1}^m\One{T(Y_{\sigma_i(1)},\dots,Y_{\sigma_i(n)})\geq T(Y_1,\dots,Y_n)}}{m+1}\textnormal{ where }\sigma_1,\dots,\sigma_m\iidsim\textnormal{Unif}(\cS_n),\]
where $\cS_n$ denotes the group of permutations on $[n]=\{1,\dots,n\}$.
Each of these constructions offers a valid p-value under the null hypothesis that $Y_1,\dots,Y_n$ are exchangeable.

In some settings, we might restrict our attention to only a certain subset of permutations: that is, given some subset $\cS_{n,\mathrm{constr}}\subseteq\cS_n$ of permutations satisfying some constraint, we might compute 
\begin{equation}\label{eqn:p_infty_constr}p_\infty = \frac{\sum_{\sigma\in\cS_{n,\mathrm{constr}}}\One{T(Y_{\sigma(1)},\dots,Y_{\sigma(n)})\geq T(Y_1,\dots,Y_n)}}{|\cS_{n,\mathrm{constr}}|},\end{equation}
or its Monte Carlo version,
\begin{equation}\label{eqn:p_m_constr}p_m = \frac{1+\sum_{i=1}^m\One{T(Y_{\sigma_i(1)},\dots,Y_{\sigma_i(n)})\geq T(Y_1,\dots,Y_n)}}{m+1}\textnormal{ where }\sigma_1,\dots,\sigma_m\iidsim\textnormal{Unif}(\cS_{n,\mathrm{constr}}).\end{equation}
For instance, this arises in methods such \emph{balanced permutations} when testing for the presence of a treatment effect, or \emph{restricted-position permutations} in settings where only certain permutations (e.g., constrained with respect to distance) are permitted (we will discuss these two examples next). 

However, these types of constrained permutation tests may lose validity, because we are computing $p_\infty$ (or $p_m$) using a subset of permutations $\cS_{n,\mathrm{constr}}\subseteq\cS_n$ that is not necessarily a subgroup \citep{hemerik2018exact}. Nonetheless, as we will see next, our theoretical guarantees can restore some Type-I error control in certain examples.

\paragraph{Notation for permutation tests of marginal independence.}
Before we discuss constrained permutation tests, we first establish some basic notation for considering the usual (unconstrained) permutation test, in the setting of testing a hypothesis of marginal independence. Consider data pairs $(A_1,Y_1),\dots,(A_n,Y_n)$ drawn i.i.d.\ from some joint distribution on $(A,Y)$, where we wish to test $H_0:A\independent Y$. Here $A$ is a covariate (e.g., a treatment assignment), while $Y$ is the response. The standard permutation test compares the dataset $\{(A_i,Y_i)\}_{i\in[n]}$ against permuted versions, $\{(A_i,Y_{\sigma(i)})\}_{i\in[n]}$, which (under $H_0$) has the same joint distribution as the original data. To align with our notation above, we now treat $A_1,\dots,A_n$ as fixed (i.e., we condition on these values), while the response values $Y_1,\dots,Y_n$ are random---and under $H_0$, the $Y_i$'s are i.i.d., even after conditioning on the $A_i$'s. Therefore, abusing notation, any test statistic that depends on $(A_1,Y_1),\dots,(A_n,Y_n)$ may be written as a function of only the response values $Y_1,\dots,Y_n$, since the $A_i$'s are treated as fixed; for instance, we might consider $T(Y_1,\dots,Y_n) = \textnormal{Corr}((A_1,\dots,A_n),(Y_1,\dots,Y_n))$, the sample correlation between $A$ and $Y$. That is, the test statistic $T$ may implicitly depend on the observed covariate values $A_1,\dots,A_n$.

With this notation in place we are now ready to examine two examples of constrained permutation tests.

\subsubsection{Restricted-position permutations for independence testing}\label{sec:dist_perm}
Although the classical permutation test for \(A\independent Y\) uses all $n!$ permutations,
many association-testing problems impose additional observational or scientific
constraints on which rearrangements are meaningful. In truncated-data problems,
for example, a response value may be moved only to positions where it would have
been observable; this leads to permutation tests over restricted positions
\citep{efron1999nonparametric,diaconis2001statistical,chen2007sequential}. In genetic
association studies with population stratification, local permutation methods
shuffle phenotypes only among ancestry-near individuals in order to avoid
unrealistic full permutations \citep{mullaert2021taking}. Spatially restricted
permutations play a similar role in ecological association testing
\citep{crabot2019testing}.

Motivated by these examples, suppose that each observation has an ancillary
design feature \(B_i\), such as a spatial coordinate, ancestry score, batch
descriptor, or observation window. 
Suppose that we only consider a permutation $\sigma$ to be permissible if it lies in the set
\[
\cS_{n,\mathrm{dist},\Delta}    =
    \{\sigma\in \cS_n:\max_i d(B_i,B_{\sigma(i)})\le \Delta\},
\]
so that each response is moved only to a nearby or admissible design position. The aim, then, is to design a permutation test that uses only permutations from this set, in order to improve scientific plausibility or power.

In general, the set $\cS_{n,\mathrm{dist},\Delta}$ need not be a subgroup of $\cS_n$. Consequently, in some settings, it may occur that the naive constrained permutation
p-value is not be valid even under a global exchangeability null \citep{hemerik2018exact}---see Example~\ref{ex:dist_worstcase} below for an instance of extreme loss of Type-I error control. To remedy this,  the following result
shows that composing two independent constrained permutations restores the
forward--backward structure and yields a finite-sample guarantee.\footnote{See also \citet[Theorems 3.1,3.2]{ramdas2023permutation} for results related to this proposition. Specifically, the result of Proposition~\ref{prop:dist_perm} for $p_\infty$ is also proved there; however the finite-sample result appears in a different form, i.e., their work proves a factor-of-$2$ bound for a different construction of $p_m$.}
\begin{proposition}\label{prop:dist_perm}
    Let $(A_1,B_1,Y_1),\dots,(A_n,B_n,Y_n)$ be i.i.d. Let $\sigma_1,\dots,\sigma_m,\sigma'_1,\dots,\sigma'_m\in\cS_{n,\mathrm{dist},\Delta/2}$ be sampled uniformly at random, and define
    \[p_m = \frac{1+\sum_{k=1}^m \One{T(Y_{\sigma_k(\sigma'_k(1))},\dots,Y_{\sigma_k(\sigma'_k(n))})\geq T(Y_1,\dots,Y_n)}}{m+1}\]
    along with its derandomized version,
    \[p_\infty = \frac{\sum_{\sigma,\sigma'\in\cS_{n,\mathrm{dist},\Delta/2}} \One{T(Y_{\sigma(\sigma'(1))},\dots,Y_{\sigma(\sigma'(n))})\geq T(Y_1,\dots,Y_n)}}{|\cS_{n,\mathrm{dist},\Delta/2}|^2},\]
    for the case $m=\infty$.
    Then, under the null $Y\independent (A,B)$,
    \[\PP{p_m\leq \alpha}\leq 2\alpha\textnormal{ for all }\alpha\in[0,1].\]
\end{proposition}
\noindent The permutation test described in this proposition is different from the constrained one above in~\eqref{eqn:p_m_constr} and~\eqref{eqn:p_infty_constr} (with the constraint set $\cS_{n,\mathrm{constr}}=\cS_{n,\mathrm{dist},\Delta}$): the permuted copies of the data are obtained by composing \emph{two} distance-constrained permutations $\sigma,\sigma'$, drawn from a smaller set $\cS_{n,\mathrm{dist},\Delta/2}$.  Nonetheless, the intuition remains the same: since the composition satisfies $\sigma\circ\sigma'\in\cS_{n,\mathrm{dist},\Delta}$, we are again running the test using only permitted permutations.

\begin{proof}[Proof of Proposition~\ref{prop:dist_perm}]
    Let $X=(Y_1,\dots,Y_n)$ be the original data, and for any permutation $\sigma$, let $X_\sigma = (Y_{\sigma(1)},\dots,Y_{\sigma(n)})$ denote the permuted data. 
    Define    
    probability kernel
    \[P(\cdot\mid x) = \frac{1}{|\cS_{n,\mathrm{dist},\Delta/2}|^2}\sum_{\sigma,\sigma'\in\cS_{n,\mathrm{dist},\Delta/2}} 
    \delta_{x_{\sigma\circ\sigma'}}.\]
    Then $p_m$ (and $p_\infty$) are constructed exactly as in~\eqref{eqn:sample_copies_P} (and~\eqref{eqn:sample_copies_P_infty}), so now we only need to verify that $P$ is a forward--backward kernel in order to apply Theorem~\ref{thm1}.
    
    In fact, this is simply an instance of Example~\ref{ex1}: we will obtain $P$ via a reversible Markov chain. Define  a Markov chain with transition kernel
    \[\pi(\cdot\mid x) = \frac{1}{|\cS_{n,\mathrm{dist},\Delta/2}|}\sum_{\sigma\in\cS_{n,\mathrm{dist},\Delta/2}} 
    \delta_{x_\sigma}.\]
    Then clearly, by construction, $P(\cdot\mid x) = \pi^2(\cdot\mid x)$, i.e., the copies are drawn by taking $s=2$ steps along the chain. Since $X=(Y_1,\dots,Y_n)$ is exchangeable (note that this holds conditionally on $A_1,\dots,A_n,B_1,\dots,B_n$), the distribution of $X$ is stationary under $\pi$. Moreover, $\sigma\in\cS_{n,\mathrm{dist},\Delta/2}$ if and only if $\sigma^{-1}\in\cS_{n,\mathrm{dist},\Delta/2}$, by definition of the constrained set, and consequently $\pi = \pi^{-1}$, i.e., the Markov chain is reversible. 
\end{proof}

Finally, we verify that without the modification proposed in the proposition, the procedure may lose Type-I error control.
\begin{example}\label{ex:dist_worstcase}
    Fix a large dimension $d\gg n^2$, and let $(A_i,B_i,Y_i)$, $i=1,\dots,n$, be i.i.d.\ copies of $(A,B,Y)$ where
    \[A = 0, \quad B \sim \frac{1}{n}\cdot\delta_{\mathbf{0}_d} + \left(1-\frac{1}{n}\right)\cdot \textnormal{Unif}(\{\mathbf{e}_1,\dots,\mathbf{e}_d\}), \quad Y\sim\textnormal{Unif}[0,1],\] with $A,B,Y$ mutually independent,
    where $\mathbf{e}_i$ is the $i$th canonical basis vector in $\R^d$. We choose the parameter $\Delta=1$ for the distance constraint, and define $T=T(Y_1,\dots,Y_n) = \prod_{i< j}(Y_i-Y_j)$.

    Let $\mathcal{E}_B$ be the event that $B_1,\dots,B_n$ are all distinct and that $B_i=\mathbf{0}_d$ for one index $i$. Then
    \[\PP{\mathcal{E}_B} = \underbrace{n \cdot\frac{1}{n}\cdot\left(1-\frac{1}{n}\right)^{n-1}}_{\substack{\textnormal{probability of exactly one $B_i=\mathbf{0}_d$}\\\textnormal{$(\approx e^{-1}$ if $n$ is large)}}}\cdot \underbrace{1\cdot\left(1-\frac{1}{d}\right)\cdot\left(1-\frac{2}{d}\right)\cdot\dots\cdot\left(1-\frac{n-2}{d}\right)}_{\substack{\textnormal{probability that the remaining $B_i$'s are distinct}\\\textnormal{($\approx 1$ if $d\gg n^2$)}}} \approx e^{-1}.\]
    On the event $\mathcal{E}_B$, due to choosing $\Delta=1$, we have
    \[\cS_{n,\mathrm{dist},\Delta} = \{\mathrm{Id}\}\cup\big\{(i,j) : j\in[n]\setminus\{i\}\big\},\]
    where $i$ is the unique index for which $B_i=\mathbf{0}_d$, and
    where $(i,j)$ denotes the permutation swapping indices $i$ and $j$.
    
    Next let $\mathcal{E}_Y$ be the event that $T>0$, which by construction has probability $\PP{\mathcal{E}_Y}=\frac{1}{2}$. On this event, by construction of the test statistic, any permutation $\sigma$ that is a swap of two indices will lead to $T(Y_{\sigma(1)},\dots,Y_{\sigma(n)})<0$. 
     Therefore,
    on the event $\mathcal{E}_B\cap \mathcal{E}_Y$ we have $T(Y_{\sigma(1)},\dots,Y_{\sigma(n)})<0$ for all $n-1$ non-identity permutations in $\cS_{n,\mathrm{dist},\Delta}$, and consequently $p_\infty = \frac{1}{n}$. Thus
    \[\PP{p_\infty \leq \frac{1}{n}} \geq \PP{\mathcal{E}_B\cap \mathcal{E}_Y} \approx \frac{1}{2e},\]
    which (for large $n$) is a gross violation of the p-value condition.
\end{example}

\subsubsection{Balanced permutation tests}\label{sec:balanced_perm}
For our second example, consider the setting of permutation testing for inference on a treatment effect. The data points are of the form $(A_i,Y_i)$, where $A_i\in\{0,1\}$ denotes a treatment assignment while $Y_i\in\cY$ represents the observed data. Assume for simplicity that there are exactly $n/2$ individuals assigned to each treatment, $A_i=1$ or $A_i=0$. Consider some test statistic $T=T(Y_1,\dots,Y_n)$, which as before is implicitly allowed to depend on $A_1,\dots,A_n$, since we will condition on the treatment assignments.

The premise of \emph{balanced permutation testing} is to restrict to permutations $\sigma$ that satisfy a balance condition:
\[\cS_{n,\mathrm{bal}} = \left\{\sigma\in\cS_n : \sum_{i=1}^n \One{A_i = a,A_{\sigma(i)}=b} =n/4\textnormal{ for each $a,b\in\{0,1\}$}\right\}\]
(where we are implicitly assuming $n$ is a multiple of $4$). Note that the subset $\cS_{n,\mathrm{bal}}$ depends implicitly on the observed treatment assignment vector $A_1,\dots,A_n$.

The idea of restricting to balanced permutations (rather than running a classical permutation test, over all $\pi\in\cS_n$) is that it may improve power to detect a treatment effect (by avoiding correlations between $Y$ and $A$ in the permuted data, as much as possible). However, it is now well-known that this approach can dramatically lose Type-I error control, under the null hypothesis that there is no treatment effect: we may even have a nontrivial probability $\PP{p_\infty = 0}$ (i.e., the event that every balanced permutation $\sigma\in \cS_{n,\mathrm{bal}}$ leads to a test statistic value that is smaller than the observed value of $T$) \citep{southworth2009properties}.

However, the following result (which applied Theorem~\ref{thm1}) demonstrates that it is nonetheless possible to bound the loss of Type-I error control.
\begin{proposition}\label{prop:balanced_perm}
    Let $n$ be a positive multiple of $4$, and let $A_1,\dots,A_n\in\{0,1\}$ be fixed, with $\sum_{i=1}^n A_i = n/2$. Let $Y_1,\dots,Y_n$ be exchangeable.
 Define $p_m$ as in~\eqref{eqn:p_m_constr} for finite $m$, or as in~\eqref{eqn:p_infty_constr} for $m=\infty$ (with the set of balanced permutations $\cS_{n,\mathrm{bal}}$ in place of $\cS_{n,\mathrm{constr}}$). Then
    \[\PP{p_m \leq \alpha}\leq 2\alpha\left(1-\frac{1}{\mathrm{Bal}_n}\right) + \frac{2}{\mathrm{Bal}_n}\textnormal{ for all $\alpha\in[0,1]$},\]
    where $\mathrm{Bal}_n$ is defined as
    \[\mathrm{Bal}_n = \max\left\{k \ : \ \exists \sigma_1,\dots,\sigma_k\in\cS_n, \ \sigma_i\circ\sigma_j^{-1}\in\cS_{n,\mathrm{bal}} \ \forall i\neq j\in[k]\right\},\]
    i.e., the maximum number of permutations such that $\sigma_i\circ\sigma_j^{-1}$ is balanced for each $i\neq j$.
\end{proposition}
\noindent See Appendix~\ref{app:balanced_perm} for the proof.

To interpret this bound, we remark that we expect $\mathrm{Bal}_n\asymp n$ (see Appendix~\ref{app:hadamard} for discussion). Taking this claim as given, this tells us that the inflation of Type-I error, for the setting of balanced permutation tests, is essentially bounded by a factor of $2$. However, at extremely small values of $\alpha$, the second term will be dominant: since $\mathrm{Bal}_n\asymp n$, values of $p_m$ that are below $\mathcal{O}(1/n)$ may not be meaningful.

\section{Proofs of main results}

In this section, we provide proofs of our main results, Theorems~\ref{thm1} and~\ref{thm2}. For both theorems, the proof will proceed by first establishing the Type-I error bound for $p_\infty$ (i.e., the case $m=\infty$), and then proving the finite-$m$ case by comparing the Monte Carlo p-value $p_m$ to its derandomized version $p_\infty$. 

\subsection{Preliminaries}
Before proving the theorems separately, we first develop a general result that will be useful for relating $p_m$ to $p_\infty$.

\subsubsection{Ordering of distributions}\label{sec:p_p*_preliminaries}
We begin by recalling some definitions for ordering of distributions: given random variables $A,B\in\R$, define the \emph{stochastic order},
\[\textnormal{$A\stleq B$ if $\EE{f(A)}\leq \EE{f(B)}$ for all nondecreasing functions $f:\R\to\R$},\]
and the \emph{convex order},
\[\textnormal{$A\cvxleq B$ if $\EE{f(A)}\leq \EE{f(B)}$ for all convex functions $f:\R\to\R$},\]
and the \emph{decreasing convex order},
\[\textnormal{$A\dcxleq B$ if $\EE{f(A)}\leq \EE{f(B)}$ for all nonincreasing and convex functions $f:\R\to\R$},\]
and the \emph{increasing convex order},
\[\textnormal{$A\icxleq B$ if $\EE{f(A)}\leq \EE{f(B)}$ for all nondecreasing and convex functions $f:\R\to\R$}.\]
where in each of these definitions, implicitly we restrict to functions $f$ for which the expected values are defined. By definition, it holds that
\[A\stgeq B \ \Longrightarrow \ A\dcxleq B\]
and
\[A\stleq B \ \Longrightarrow \ A\icxleq B\]
and
\[A\cvxleq B \ \Longrightarrow \ A\dcxleq B\textnormal{ and }A\icxleq B.\]

For a random variable $p\in[0,1]$, if $p\stgeq U$ where $U\sim\textnormal{Unif}[0,1]$, we say that $p$ is \emph{superuniform} (and might also write $p\stgeq\textnormal{Unif}[0,1]$). In this case, $p$ is a valid p-value, i.e., $\PP{p\leq\alpha}\leq\alpha$ for all $\alpha$. On the other hand, if $p\dcxleq U$, \citet{wang2024testing} call $p$ a `p*-variable', and establish that
\begin{equation}\label{eqn:p*_factor_of_2}\textnormal{For any p*-variable $p$, $\PP{p\leq\alpha}\leq2\alpha$ for all $\alpha$.}\end{equation} An equivalent condition is that $p\geq\EEst{U}{V}$ almost surely for some random variables $U,V$ where $U\sim\textnormal{Unif}[0,1]$. This result is also related to the work of \citet{ruschendorf1982random,meng1994posterior,vovk2020combining}, which establish that \emph{an average of p-values is a valid p-value up to a factor of $2$}. We will use these results throughout our proofs below.

\subsubsection{Comparing binomials with random parameters}
Next,
recall the definition~\eqref{eqn:sample_copies_P_infty} of $p_\infty$, and note that $p_\infty$ is a function of $X$ (and is therefore random). Under this definition, we can see that the indicator variables 
\[\One{T(X_i)\geq T(X)}, \quad i=1,\dots,m\]
are i.i.d.\ draws from the $\textnormal{Bernoulli}(p_\infty)$ distribution, conditional on $X$. We can therefore write
\begin{equation}\label{eqn:pm_Binomial}p_m = \frac{B+1}{m+1} \textnormal{ where } B \mid X \sim \textnormal{Binomial}(m,p_\infty).\end{equation}
With this calculation in place, the following proposition allows us to characterize how properties of the derandomized p-value $p_\infty$ may be inherited by its Monte Carlo version $p_m$.

\begin{proposition}\label{prop:ordering}
    Let $V_0,V_1\in[0,1]$ be random variables, and let
    \[B_i\mid V_i \sim \textnormal{Binomial}(m,V_i),\]
    for each $i=0,1$. Then it holds that
    \[\textnormal{If $V_0\preceq V_1$ then $B_0\preceq B_1$},\]
    where $\preceq$ may denote either $\stleq$, $\cvxleq$, $\dcxleq$, or $\icxleq$.
\end{proposition}
\noindent This result is proved in Appendix~\ref{app:prove_prop:ordering}.
We will also need an additional result:
\begin{lemma}\label{lem:dataprocessing}
    Let $V_0,V_1,B_0,B_1$ be defined as in Proposition~\ref{prop:ordering}. Then it holds that
    \[\dtv(B_0,B_1)\leq\dtv(V_0,V_1).\]
\end{lemma}
\begin{proof}[Proof of Lemma~\ref{lem:dataprocessing}]
This result is simply a consequence of the data processing inequality \citep{liese2006divergences}.    
\end{proof}
Finally, we need one more result for the case where the random parameter is uniform:
\begin{lemma}\label{lem:binomial_uniform}
    If $V\sim\textnormal{Unif}[0,1]$ and $B\mid V\sim\textnormal{Binomial}(m,V)$, then 
    \[B \sim \textnormal{Unif}(\{0,1,\dots,m\}).\]
\end{lemma}
\begin{proof}[Proof of Lemma~\ref{lem:binomial_uniform}]
    Let $V_1,\dots,V_m\iidsim\textnormal{Unif}[0,1]$ be drawn independently of $V$. Then $B\eqd \sum_{i=1}^m \One{V_i\leq V}$. Moreover, since $V,V_1,\dots,V_m$ are i.i.d.\ and therefore exchangeable, the rank of $V$ among this list is uniform, and so $B \sim \textnormal{Unif}(\{0,1,\dots,m\})$ as desired.
\end{proof}

\subsection{Proof of Theorem~\ref{thm1}}
In this section we prove a stronger claim: we will show that $p_m$ (for finite $m$ and for $m=\infty$) is a p*-variable, i.e., $p_m\dcxleq \textnormal{Unif}[0,1]$. By~\eqref{eqn:p*_factor_of_2}, this will immediately imply that $p_m$ is valid up to a factor of $2$ \citep{wang2024testing}, yielding the result of Theorem~\ref{thm1}.

\paragraph{Infinite case.} 
First we consider the case $m=\infty$. 
Recalling that $P$ is a forward--backward kernel as defined in~\eqref{eqn:fb_condition}, we can take a joint distribution on the triple $(X,X',Y)$ such that:
\begin{itemize}
    \item The marginal distribution of $X$ is $Q_0$;
    \item The conditional distribution of $X'\mid X$ is $P(\cdot\mid X)$;
    \item And, $X,X'$ are conditionally i.i.d.\ given $Y$.
\end{itemize}
We will use the following standard fact: 
\begin{fact}\label{fact:conditional_iid_pval}
    If $X,X'$ are conditionally i.i.d.\ given $Y$, and $T$ is any function, then
    \[p = \PPst{T(X')\geq T(X)}{X,Y}\]
    is a valid p-value conditional on $Y$:
    \[\PPst{p\leq \alpha}{Y}\leq \alpha\textnormal{ for all $\alpha\in[0,1]$, i.e., $p\stgeq \textnormal{Unif}[0,1]$.}\]
    (In words, $p$ is the p-value for test statistic $T(X)$, with respect to the distribution of $X\mid Y$.)
\end{fact}

Now define
\[\tilde{p}_\infty = \PPst{T(X')\geq T(X)}{X,Y},\]
which therefore satisfies 
\begin{equation}\label{eqn:avg_of_pval_1}\tilde{p}_\infty\stgeq \textnormal{Unif}[0,1].\end{equation}
by Fact~\ref{fact:conditional_iid_pval}.

Next, it holds by definition that $p_\infty = \PPst{T(X')\geq T(X)}{X}$, under the joint distribution of $(X,X',Y)$. 
Consequently, by the tower law, we also have
\begin{equation}\label{eqn:avg_of_pval_2}p_\infty = \EEst{\tilde{p}_\infty}{X}.\end{equation}
Combining everything, we have shown that $p_\infty$ can be represented as a conditional expectation of $\tilde{p}_\infty$~\eqref{eqn:avg_of_pval_2}, which implies $p_\infty\cvxleq\tilde{p}_\infty$, and  moreover $\tilde{p}_\infty$ is itself superuniform~\eqref{eqn:avg_of_pval_1}; combining these facts yields $p_\infty\dcxleq \textnormal{Unif}[0,1]$, as desired.

\paragraph{Finite case.} Now we are ready to turn to the finite-$m$ case. 
Define Binomial random variables
\[B_0\mid p_\infty \sim\textnormal{Binomial}(m,p_\infty),\]
and
\[B_1\mid U\sim \textnormal{Binomial}(m,U),\]
where $U\sim\textnormal{Unif}[0,1]$.
Recall from~\eqref{eqn:pm_Binomial} that we can write $p_m = \frac{B_0+1}{m+1}$. By Proposition~\ref{prop:ordering}, we have
\[p_\infty\dcxleq U \ \Longrightarrow \ B_0 \dcxleq B_1\  \Longleftrightarrow \ p_m  \dcxleq \frac{B_1 + 1}{m+1}.\]
But by Lemma~\ref{lem:binomial_uniform}, we have $B_1 \sim\textnormal{Unif}(\{0,\dots,m\})$. Therefore $\frac{B_1 + 1}{m+1}\sim\textnormal{Unif}(\{\frac{1}{m+1},\dots,\frac{m}{m+1},1\})$ and so
\[\frac{B_1 + 1}{m+1}\stgeq \textnormal{Unif}[0,1] \ \Longrightarrow \ \frac{B_1 + 1}{m+1}\dcxleq\textnormal{Unif}[0,1].\]
Combining everything, we have shown that
$p_m  \dcxleq \textnormal{Unif}[0,1]$,
as desired.

\subsection{Proof of Theorem~\ref{thm2}}
\paragraph{Infinite case.} 
First we consider the case $m=\infty$. Let $X,X',X''$ be distributed as in~\eqref{eqn:P_tv_condition}. Then
by construction,
\[p_\infty 
=\PPst{T(X')\geq T(X)}{X},\]
and we also define
\[\tilde{p}_\infty = \PPst{T(X'')\geq T(X)}{X}.\]
We can write
\[\tilde{p}_\infty = \PPst{T(X'')\geq T(X)}{X}= \PPst{T(X'')\geq T(X)}{X,f(X)}.\]
And, $X,X''$ are conditionally i.i.d.\ given $f(X)$, since each has the conditional distribution $Q_0(\cdot\mid f(X))$ by definition. Therefore by Fact~\ref{fact:conditional_iid_pval} (applied with $Y=f(X)$), we have $\tilde{p}_\infty\stgeq\textnormal{Unif}[0,1]$. 

Moreover, 
by construction, conditional on $X$ it holds almost surely that
\[|p_\infty - \tilde{p}_\infty|\leq \dtv\big(X'\mid X, X''\mid X\big),\]
i.e., the total variation distance between the two conditional distributions. Therefore, under the assumption~\eqref{eqn:P_tv_condition}, we have
\[\EE{|p_\infty - \tilde{p}_\infty|} \leq \EE{\dtv\big(X'\mid X, X''\mid X\big)} = \dtv((X,X'),(X,X'')) \leq \epsilon.\]
    Next we need a lemma:\footnote{We remark that this lemma is similar to existing bounds in the literature which could be used to address the case where $U$ is exactly uniform rather than superuniform---e.g., \citet[Proposition 1.2 part 2]{ross2011fundamentals}, which bounds Kolmogorov--Smirnov distance via Wasserstein distance.}
\begin{lemma}\label{lem:dtv_unif}
    Let $U,V\in[0,1]$ be random variables, where $U\stgeq \textnormal{Unif}[0,1]$. Then
    \[\PP{V\leq a}\leq a + \sqrt{2\EE{(U-V)_+}} \textnormal{ for all $a\in[0,1]$}.\]
\end{lemma}
\noindent
Combining everything, and applying the lemma with $U=\tilde{p}_\infty$ and $V=p_\infty$, we have therefore shown that
\[\PP{p_\infty\leq \alpha}\leq  \alpha + \sqrt{2\epsilon},\]
as desired.

Next suppose instead that the stronger total variation condition~\eqref{eqn:P_tv_condition_strong} holds. Then, conditional on $X$, the total variation distance between the conditional distributions of $X'$ and of $X''$ is at most $\epsilon$; consequently we have
\[|p_\infty-\tilde{p}_\infty|\leq \epsilon\textnormal{ almost surely}.\]
Therefore,
\[\PP{p_\infty\leq \alpha}\leq \PP{\tilde{p}_\infty\leq \alpha+\epsilon}\leq \alpha+\epsilon,\]
where the last step holds since $\tilde{p}_\infty$ is superuniform.

\paragraph{Finite case.} 
Let $\Delta = \sqrt{2\epsilon}$ if we are working under the assumption~\eqref{eqn:P_tv_condition}, or $\Delta=\epsilon$ if we are working under the stronger condition~\eqref{eqn:P_tv_condition_strong}. Without loss of generality we can assume $\Delta\leq1$, otherwise the result is trivial.

Let $U\sim\textnormal{Unif}[0,1]$, and consider the random variable
\[(U-\Delta)_+ = \max\big\{U-\Delta,0\big\},\]
which has cumulative distribution function
\[F(t) = \PP{(U-\Delta)_+\leq t} = \min\{t + \Delta , 1\}, \quad t\in[0,1].\]
The bound proved above for $p_\infty$ therefore establishes that
\[p_\infty\stgeq (U-\Delta)_+.\]
Now let
\[B_0\mid p_\infty \sim\textnormal{Binomial}(m,p_\infty), \ B_1\mid U \sim\textnormal{Binomial}\big(m,(U-\Delta)_+\big), \ B_2\mid U \sim\textnormal{Binomial}(m,U).\]
By Proposition~\ref{prop:ordering}, we therefore have
\[p_\infty\stgeq (U-\Delta)_+ \ \Longrightarrow \ B_0 \stgeq B_1 \ \Longleftrightarrow \ p_m \stgeq \frac{B_1+1}{m+1},\]
since $p_m = \frac{B_0+1}{m+1}$. 
And, by Lemma~\ref{lem:dataprocessing},
\[\dtv(B_1,B_2)\leq \dtv\big((U-\Delta)_+,U\big) = \Delta.\]
Combining everything,
\begin{align*}
    \PP{p_m\leq \alpha}
    &\leq \PP{\frac{B_1+1}{m+1}\leq \alpha}\textnormal{\quad since $p_m\stgeq \frac{B_1+1}{m+1}$}\\
&\leq \PP{\frac{B_2+1}{m+1}\leq \alpha} + \Delta\textnormal{\quad since $\dtv(B_1,B_2)\leq\Delta$}\\
    &\leq \alpha+\Delta\textnormal{\quad since $\frac{B_2+1}{m+1}\stgeq \textnormal{Unif}[0,1]$ by Lemma~\ref{lem:binomial_uniform}.}
\end{align*}

\section{Discussion}

We have studied a simple but subtle question: what can be guaranteed when
Monte Carlo copies are generated from a null-compatible sampling kernel, but the
resulting collection is not jointly exchangeable with the observed data?  The
classical empirical p-value is exactly valid when the observed data and the
Monte Carlo copies are jointly exchangeable.  The examples in the introduction
show that pairwise exchangeability, or stationarity of an MCMC kernel, is not
enough by itself: the naive parallel sampler can have severely inflated type-I
error when the chain has not mixed.  At the same time, the fully exchangeable
hub-and-spoke construction of \citet{besag1989generalized} may introduce substantial
conditional Monte Carlo randomness, because all spokes share the same latent hub.

Our main results identify an intermediate regime.  Under the
forward--backward condition, the derandomized quantity
\[
    p_\infty = \mathbb P\{T(X')\ge T(X)\mid X\}
\]
is a p$^*$-value, and the same p$^*$-value guarantee is inherited by its finite
Monte Carlo version
\[
    p_m=\frac{1+\sum_{i=1}^m {\bf 1}\{T(X_i)\ge T(X)\}}{m+1}.
\]
Consequently, rejecting at threshold $\alpha/2$ yields a valid level-$\alpha$
test without any mixing assumption.  This guarantee is weaker than exact
validity, but it is nonasymptotic in both the number of Monte Carlo samples and
the number of Markov chain steps.  When the sampling kernel has mixed, or more
generally when it is close in total variation to an appropriate conditional null
kernel, the factor-of-$2$ multiplicative guarantee improves to an additive $\alpha+\epsilon$-type bound.
Thus the two main phenomena are complementary: forward--backward structure
protects against arbitrarily poor mixing, while mixing recovers nominal
validity.

\paragraph{Relation to MCMC significance tests.}
Our results are closely related to the generalized Monte Carlo significance
tests of Besag and Clifford.  Their parallel
hub-and-spoke method constructs copies that are jointly exchangeable with the
observed data, and therefore gives an exactly valid p-value for any fixed test
statistic.  More recent work by \citet{Howes2026} reviews these MCMC
significance tests, including serial variants, and presents a unifying
exchangeability perspective.  Our focus is different.  We analyze the naive
parallel sampler, or slight modifications of it, in settings where the observed
data and the copies are typically not jointly exchangeable.  For reversible
chains, an even number of steps can be written as a forward--backward move, so
the naive sampler enjoys a factor-of-$2$ type-I guarantee even if the chain is far
from mixed.  For non-reversible chains, analogous guarantees can be recovered by
using forward--backward compositions such as $\pi^r\circ\pi^{-r}$ or
$(\pi^{-1}\circ\pi)^r$.

\paragraph{Relation to posterior predictive checks.}
The forward--backward condition also clarifies the behavior of Bayesian
posterior predictive p-values.  Under a Bayesian model with prior $\nu$ and
likelihoods $\{f_\theta\}$, drawing
\[
    \theta\mid X,\qquad X'\mid \theta
\]
is exactly a forward--backward construction, with the latent parameter
$\theta$ playing the role of the intermediate variable.  Our Theorem~1 therefore
recovers, and extends to finite Monte Carlo estimates, the familiar factor-of-$2$
validity of posterior predictive p-values \citep{Rubin1984,meng1994posterior}.  This
connects our work to posterior predictive model checking
\citep{GelmanMengStern1996} and to calibration or post-processing approaches
that transform posterior predictive p-values to a uniform scale
\citep{HjortDahlSteinbakk2006}.  Those calibration methods pursue exact
uniformity, often through an additional layer of simulation; our results instead
give direct nonasymptotic guarantees for the original Monte Carlo comparison.

\paragraph{Other methods based on exchangeable copies.}
Many classical and modern testing procedures can be viewed as generating copies
of the observed data under the null.  Permutation tests
\citep{pitman1937significance,fisher1956mathematics} obtain exact validity from invariance of the null
under a group action.  Co-sufficient sampling conditions on a sufficient
statistic in a parametric model and then resamples from the conditional
distribution; approximate co-sufficient sampling extends
this idea by conditioning on approximately sufficient statistics, yielding
approximately exchangeable copies and finite-sample inflation bounds
\citep{BarberJanson2022,ZhuBarber2023,bhaduri2026conditioning}.  In model-X conditional independence
testing, the conditional randomization test resamples covariates from a known or
estimated conditional distribution \citep{CandesFanJansonLv2018}, while the
conditional permutation test uses a non-uniform distribution over permutations to
respect the dependence between the tested covariate and the confounders
\citep{BerrettWangBarberSamworth2020}.  These methods generally aim to create
jointly exchangeable, or approximately jointly exchangeable, copies.  In
contrast, our forward--backward results show that useful type-I guarantees can
remain even when only a weaker, one-copy exchangeability structure is available.

\paragraph{Limitations.}
The price of weakening joint exchangeability is a restriction on the test
statistic.  In a fully exchangeable construction, one may use any statistic that
is computed symmetrically from the observed data and all generated copies.  This
allows, for example, refitting a model or reselecting tuning parameters after
each permutation or resampling step.  Our results instead apply to a
prespecified statistic $T:\mathcal X\to\mathbb R$ evaluated separately on the
observed data and on each copy.  This distinction is analogous to the difference
between full conformal and split conformal inference: full exchangeability
permits symmetric retraining across all candidate samples, whereas
split-conformal or training-conditional guarantees require the score function to
be fixed before the calibration step.

\paragraph{Approximate null compatibility and bootstrap-type procedures.}
Another important extension is to sampling mechanisms that are not exactly
compatible with the null.  Parametric bootstrap and plug-in resampling
procedures often replace an unknown null distribution by an estimated one; these
procedures are typically valid only asymptotically and can fail in finite
samples.  
The total-variation bounds in Theorem~2 suggest a route to finite-sample
robustness guarantees for plug-in and bootstrap samplers, whenever the
implemented sampler can be compared sharply to an ideal null-compatible sampler.
Finding ways to obtain sharp characterizations of the accuracy of common
bootstrap and simulation-based inference procedures remains an open problem.

\paragraph{Serial sampling and adaptive computation.}
We have focused on parallel Monte Carlo copies sampled conditionally
independently from a kernel $P(\cdot\mid X)$.  Besag and Clifford also proposed
serial MCMC significance tests, where the copies arise along a single Markov
chain trajectory.  Serial sampling may be computationally preferable, and may
reduce or increase Monte Carlo variability depending on the chain and statistic.
An interesting open question is whether analogues of our forward--backward and
total-variation guarantees can be proved for serial samplers, possibly under
weaker dependence conditions than conditional independence of the copies.
A second open question is how far the prespecified-statistic requirement can be
relaxed.  Allowing data-adaptive choices of $T$, while preserving a
nonasymptotic guarantee, would substantially broaden the practical scope of the
method.

Overall, the message is that exact joint exchangeability is sufficient but not
necessary for useful Monte Carlo inference.  Forward--backward structure,
approximate mixing, and finite exchangeability each provide different ways to
control the price paid for using computationally convenient samples that are not
fully exchangeable with the data.

\section*{Acknowledgements} R.F.B. was partially supported by the Office of Naval Research via grant N00014-24-1-2544.
 The authors thank 
Rohan Hore, Art Owen, and Yuling Yao for helpful discussions.

\bibliographystyle{plainnat}
\bibliography{bib}

\appendix

\section{Additional proofs}

\subsection{Proof of Proposition~\ref{prop:ordering}}\label{app:prove_prop:ordering}
First, we restate this part of the proposition in a more general form. 
Given a class $\Fcal$ of functions $f:\R\to\R$, define $A\Fleq B$ if $\EE{f(A)}\leq \EE{f(B)}$ for all $f\in\Fcal$ for which these expected values are defined.
Then the four orderings $\stleq$, $\cvxleq$, $\dcxleq$, and $\icxleq$ can be obtained by choosing $\Fcal$ as, respectively, the set of all nondecreasing functions,  the set of all convex functions, the set of all nonincreasing and convex functions, or the set of all nondecreasing and convex functions.

\begin{proposition}\label{prop:ordering_F}
    Let $V_0,V_1\in[0,1]$ be random variables, and let
    \[B_i\mid V_i \sim \textnormal{Binomial}(m,V_i),\]
    for each $i=0,1$. Let $\Fcal$ be a set of functions $f:\R\to\R$ such that
    \begin{equation}\label{eqn:Fcal_condition}
        \textnormal{For all $f\in\Fcal$, there exists a $g\in\Fcal$ with $g(p)=\EE{f(\textnormal{Binomial}(m,p))}$ for all $p\in[0,1]$.}
    \end{equation}
    Then it holds that
    \[\textnormal{If $V_0\Fleq V_1$ then $B_0\Fleq B_1$}.\]
\end{proposition}

\begin{proof}[Proof of Proposition~\ref{prop:ordering}]
In order to obtain this result as corollary of Proposition~\ref{prop:ordering_F}, we only need to verify that the condition~\eqref{eqn:Fcal_condition} holds for each of the four relevant choices of $\Fcal$. 

Fix any $f$ and define $g(p) =\EE{f(\textnormal{Binomial}(m,p))}$.
Then it holds 
    \begin{equation}\label{eqn:f_g_1}\textnormal{If $f$ is nondecreasing (or, nonincreasing) then $g$ is nondecreasing (or, nonincreasing),}\end{equation}
    trivially since $\textnormal{Binomial}(m,p)\stleq \textnormal{Binomial}(m,p')$ for $p\leq p'$, and moreover
    \begin{equation}\label{eqn:f_g_2}\textnormal{If $f$ is convex then $g$ is convex},\end{equation}
    which is a classical fact about exponential families \citep{shaked1980mixtures,schweder1982dispersion}.

However, we have not yet completed the proof, since at the moment $g$ is defined on the domain $p\in[0,1]$. We now show that $g$ can be extended to a function on $\R$ while preserving the above properties.

First we calculate a one-sided derivative for $g$ at $p=0$,
\begin{align*}
    g'(0)
    &=\lim_{p\searrow 0}\frac{g(p)-g(0)}{p} \\
    &=\lim_{p\searrow 0}\frac{\EE{f(\textnormal{Binomial}(m,p))} - \EE{f(\textnormal{Binomial}(m,0))}}{p}\\
    &=\lim_{p\searrow 0} \frac{\sum_{i=0}^m \binom{m}{i}p^i(1-p)^{m-i} \big(f(i) - f(0)\big)}{p} =m\big(f(1)-f(0)\big).
\end{align*}
And similarly at $p=1$,
\[g'(1) = \lim_{p\nearrow 1}\frac{g(1)-g(p)}{1-p} = m\big(f(m)-f(m-1)\big).\]
Since $g$ is differentiable on $[0,1]$, we can therefore define an extension to a function $g:\R\to\R$ (by taking $g'(t) = g'(0)$ for $t<0$, and $g'(t)=g'(1)$ for $t>1$) such that
\begin{multline*}\textnormal{If $g$ is nondecreasing (or nonincreasing) on $[0,1]$,}\\\textnormal{then $g$ is nondecreasing (or nonincreasing) on $\R$},\end{multline*}
and such that
\[\textnormal{If $g$ is convex on $[0,1]$, then $g$ is convex on $\R$}.\]
Therefore, the claims~\eqref{eqn:f_g_1} and~\eqref{eqn:f_g_2} hold for the extended function $g:\R\to\R$. Consequently, 
\[f\in\Fcal \ \Longrightarrow \ g\in\Fcal\]
for each of the relevant choices of $\Fcal$ (i.e., nondecreasing functions; convex functions; nonincreasing convex functions; nondecreasing convex functions), which completes the proof.
\end{proof}

\begin{proof}[Proof of Proposition~\ref{prop:ordering_F}]
    Fixing any $f\in\Fcal$, we need to show that $\EE{f(B_0)}\leq \EE{f(B_1)}$ (assuming  these expected values exist). By assumption~\eqref{eqn:Fcal_condition}, we can fix some $g\in\Fcal$ with
    \[g(p) = \EE{f(\textnormal{Binomial}(m,p))}\]
    for all $p\in[0,1]$. And, for each $i=0,1$, we have
    \[g(V_i) = \EEst{f(\textnormal{Binomial}(m,V_i))}{V_i} = \EEst{f(B_i)}{V_i},\]
    and so
    \(\EE{g(V_i)} =  \EE{f(B_i)}.\)
    Therefore, 
    \[\EE{f(B_0)} = \EE{g(V_0)}\leq \EE{g(V_1)} = \EE{f(B_1)},\]
    where the inequality holds since $g\in\Fcal$ and $V_0\Fleq V_1$.
\end{proof}

\subsection{Proof of Lemma~\ref{lem:dtv_unif}}

    Without loss of generality, we can assume that $U\sim\textnormal{Unif}[0,1]$, since this is the most challenging case (otherwise we may replace $U$ with a uniform, $U'\sim\textnormal{Unif}[0,1]$ with $U'\leq U$ almost surely, so that $(U'-V)_+\leq (U-V)_+$ almost surely).

    First, let $h:[0,1]\to[0,1]$ be any function, such that $h(u)=0$ if $u\leq a$. Let $b=\EE{h(U)}$. It then holds that
    \[\EE{U\cdot h(U)} \geq ab + b^2/2,\]
    with the lower bound attained by choosing $h(u)=\one{a<u\leq a+b}$.

    Now we choose the function as
    \[h(U) = \PPst{V\leq a}{U} \cdot \One{U>a}.\]
    Again let $b=\EE{h(U)} = \PP{V\leq a, U>a}$. Then
    \begin{align*}
        ab + b^2/2
        &\leq \EE{U\cdot h(U)} =  \EE{U \cdot \One{V\leq a, U>a}}\\
        &\leq   \EE{\big(a + (U-V)_+\big) \cdot \One{V\leq a, U>a}}\\
        &\leq a \cdot \PP{V\leq a, U>a} + \EE{(U-V)_+} = ab + \EE{(U-V)_+}.
    \end{align*}
    This proves that $ b\leq \sqrt{2\EE{(U-V)_+}}$.
    Therefore,
    \[\PP{V\leq a}\leq \PP{U\leq a} + \PP{V\leq a, U>a} \leq a +  \sqrt{2\EE{(U-V)_+}}.\]

\subsection{Proofs and additional details for balanced permutation tests}\label{app:balanced_perm}

\subsubsection{Proof of Proposition~\ref{prop:balanced_perm}}

First we rewrite the problem in our general notation. Let $X=(Y_1,\dots,Y_n)$ denote the data (note that $A_1,\dots,A_n$ are fixed and so we treat these values as constants).
    
    Define a probability kernel
    \[P(\cdot\mid X) = \frac{1}{|\cS_{n,\mathrm{bal}}|}\sum_{\sigma\in \cS_{n,\mathrm{bal}}}\delta_{X_\sigma}\textnormal{ where }X_\sigma = (Y_{\sigma(1)},\dots,Y_{\sigma(n)}).\]
    In other words, we are sampling a permutation $\sigma$ uniformly at random from the set of balanced permutations $\cS_{n,\mathrm{bal}}$, and then returning the permuted version of the data.
    Under this notation, we can see that the quantities $p_m$~\eqref{eqn:sample_copies_P} and $p_\infty$~\eqref{eqn:sample_copies_P_infty} defined for our general framework coincide exactly with the quantities $p_m$ and $p_\infty$ constructed in Section~\ref{sec:balanced_perm} for the balanced permutation test.
    
    Next, define also a probability kernel
    \[P^*(\cdot\mid X) = \left(1-\frac{1}{\mathrm{Bal}_n}\right) \cdot P(\cdot \mid X)+\frac{1}{\mathrm{Bal}_n} \cdot \delta_X.\]
    Below, we will verify that $P^*$ is a forward--backward probability kernel, i.e., it satisfies the condition~\eqref{eqn:fb_condition}. Consequently, if we define
    \[p^*_\infty = \Ppst{X'\sim P^*(\cdot\mid X)}{T(X')\geq T(X)}{X},\]
    then, as in Theorem~\ref{thm1}, we have $p^*_\infty\dcxleq \textnormal{Unif}[0,1]$. 
    But by construction, we have
    \[p^*_\infty =\left(1-\frac{1}{\mathrm{Bal}_n}\right) \cdot p_\infty+\frac{1}{\mathrm{Bal}_n}.\]
    In particular, we then have
    \[\PP{p_\infty\leq\alpha} = \PP{p^*_\infty\leq \alpha \left(1- \frac{1}{\mathrm{Bal}_n}\right) +\frac{1}{\mathrm{Bal}_n}}\leq 2\alpha\left(1-\frac{1}{\mathrm{Bal}_n}\right)+\frac{2}{\mathrm{Bal}_n}\]
    by~\eqref{eqn:p*_factor_of_2}.
    
    Next, for the finite case, we need another result (proved in Appendix~\ref{app:proof_prop:binomial_concave}):
    \begin{proposition}\label{prop:binomial_concave}
       Let $V\in[0,1]$ be a random variable whose CDF $F(t) = \PP{V\leq t}$ is a concave function on $t\in[0,1]$. Let $B\mid V\sim\textnormal{Binom}(m,V)$. Then
       \[\frac{B+1}{m+1}\stgeq V.\]
    \end{proposition}
    
    Now let $V\in[0,1]$ be a random variable with CDF $F(t) = \min\left\{2t\left(1-\frac{1}{\mathrm{Bal}_n}\right)+\frac{2}{\mathrm{Bal}_n} , 1\right\}$ on $t\in[0,1]$, which is concave. Then by comparing CDFs, we see that $p_\infty\stgeq V$. We also have
    \[p_m \stgeq \frac{B+1}{m+1},\]
    where $B\mid V\sim\textnormal{Binom}(m,V)$, by Proposition~\ref{prop:ordering}. Consequently, applying Proposition~\ref{prop:binomial_concave},
    \[p_m\stgeq V\]
    and so
    \[\PP{p_m\leq \alpha}\leq \PP{V\leq \alpha} = F(\alpha) \leq 2\alpha\left(1-\frac{1}{\mathrm{Bal}_n}\right)+\frac{2}{\mathrm{Bal}_n},\]
    as desired.

    To complete the proof, we return to the question of verifying that $P^*$ is a forward--backward kernel. Define
    \[\tilde\cS_{n,\mathrm{bal}} = \left\{(\sigma_1,\dots,\sigma_{\mathrm{Bal}_n}) : \textnormal{ $\sigma_i\circ\sigma_j^{-1}\in\cS_{n,\mathrm{bal}}$ for all $i\neq j\in\{1,\dots,\mathrm{Bal}_n\}$}\right\}.\]
    By definition of $\mathrm{Bal}_n$, this set is nonempty. Moreover, by construction it satisfies certain symmetry conditions. First, if $(\sigma_1,\dots,\sigma_{\mathrm{Bal}_n})\in\tilde\cS_{n,\mathrm{bal}}$ then any permutation of this vector $(\sigma_{\tau(1)},\dots,\sigma_{\tau(\mathrm{Bal}_n)})$ (where $\tau$ is a permutation of $\{1,\dots,\mathrm{Bal}_n\}$) is also in $\tilde\cS_{n,\mathrm{bal}}$, and consequently
    \[\textnormal{If $(\sigma_1,\dots,\sigma_{\mathrm{Bal}_n})\sim\textnormal{Unif}(\tilde\cS_{n,\mathrm{bal}})$ then $\sigma_1,\dots,\sigma_{\mathrm{Bal}_n}$ are exchangeable.}\]
    Second,
    \[\textnormal{If $(\sigma_1,\dots,\sigma_{\mathrm{Bal}_n})\sim\textnormal{Unif}(\tilde\cS_{n,\mathrm{bal}})$ then $\sigma_i\circ\sigma_j^{-1} \sim\textnormal{Unif}(\cS_{n,\mathrm{bal}})$},\]
    for each $i\neq j\in\{1,\dots,\mathrm{Bal}_n\}$.
    Therefore, the probability kernel $P$ can equivalently be written as 
    \[P(\cdot\mid X) = \frac{1}{|\tilde\cS_{n,\mathrm{bal}}|}\sum_{(\sigma_1,\dots,\sigma_{\mathrm{Bal}_n})\in\tilde\cS_{n,\mathrm{bal}}} \frac{1}{\mathrm{Bal}_n(\mathrm{Bal}_n-1)}\sum_{i\neq j\in\{1,\dots,\mathrm{Bal}_n\}}\delta_{X_{\sigma_i\circ\sigma_j^{-1}}}.\]
    By definition of $P^*$, then,
    \[P^*(\cdot\mid X) = \frac{1}{|\tilde\cS_{n,\mathrm{bal}}|}\sum_{(\sigma_1,\dots,\sigma_{\mathrm{Bal}_n})\in\tilde\cS_{n,\mathrm{bal}}} \frac{1}{\mathrm{Bal}_n^2}\sum_{i,j\in\{1,\dots,\mathrm{Bal}_n\}}\delta_{X_{\sigma_i\circ\sigma_j^{-1}}}.\]
    Now define a random variable $Z = (\sigma_1,\dots,\sigma_{\mathrm{Bal}_n})$, and let $i,j\in\{1,\dots,\mathrm{Bal}_n\}$ be sampled uniformly at random (with replacement). Since $X$ is exchangeable (and so $X\eqd X_{\sigma_i^{-1}}$, even after conditioning on $Z,i,j$),
    \[\big(X,X_{\sigma_i\circ\sigma_j^{-1}},Z\big) \eqd \big(X_{\sigma_i^{-1}},(X_{\sigma_i^{-1}})_{\sigma_i\circ\sigma_j^{-1}},Z\big)=\big(X_{\sigma_i^{-1}},X_{\sigma_j^{-1}},Z\big).\]
    And clearly, $X_{\sigma_i^{-1}},X_{\sigma_j^{-1}}$ are conditionally i.i.d.\ given $(Z,X)$, since the indices $i,j$ are sampled uniformly with replacement. Therefore, this verifies that $P^*$ satisfies the forward--backward condition~\eqref{eqn:fb_condition}.
    
\subsubsection{Proof of Proposition~\ref{prop:binomial_concave}}\label{app:proof_prop:binomial_concave}
First we prove the result for a special case. Let $V\sim\textnormal{Unif}[0,a]$ where $a\in[0,1]$. 
Then, for $t\in[0,1]$, if $a>0$ we have
\begin{align*}
    \PP{\frac{B+1}{m+1}\leq t}
    &=\Pp{V\sim\textnormal{Unif}[0,a]}{\frac{B+1}{m+1}\leq t}\\
    &=\Ppst{V\sim\textnormal{Unif}[0,1]}{\frac{B+1}{m+1}\leq t}{V\leq a}\\   
    &\leq \frac{\Pp{V\sim\textnormal{Unif}[0,1]}{\frac{B+1}{m+1}\leq t}}{\Pp{V\sim\textnormal{Unif}[0,1]}{V\leq a}} \leq \frac{t}{a},
\end{align*}
where the last step holds since, for $V\sim\textnormal{Unif}[0,1]$, the random variable $\frac{B+1}{m+1}$ is superuniform, by Lemma~\ref{lem:binomial_uniform}. Consequently $\frac{B+1}{m+1}\stgeq V$. On the other hand if $a=0$, then $V=0$ almost surely and so $\frac{B+1}{m+1}\stgeq V$ holds trivially.

Next we move to the general case.
Since the CDF $F$ is a concave function, we can write
\[V \eqd AU\]
where $U\sim\textnormal{Unif}[0,1]$ and $A$ is a random variable with $A\independent U$, by Khintchine's theorem \citep[Chapter V.9]{feller1991introduction}. Since $V$ takes values in $[0,1]$ we must have $A\in[0,1]$ almost surely as well.
Then, by the work above,
\[\left(\frac{B+1}{m+1}\mid A=a\right) \stgeq a\cdot U\]
for all $a\in[0,1]$. Therefore,
\begin{multline*}
    \PP{\frac{B+1}{m+1}\leq t}
    =\EE{\PPst{\frac{B+1}{m+1}\leq t}{A}}\\
    \leq \EE{\PPst{A\cdot U\leq t}{A}}
    =\PP{A\cdot U\leq t} = \PP{V\leq t},
\end{multline*}
proving that $\frac{B+1}{m+1}\stgeq V$.

\subsubsection{Additional calculations: computing $\mathrm{Bal}_n$}\label{app:hadamard}
In Section~\ref{sec:balanced_perm}, to help interpret the result of Proposition~\ref{prop:balanced_perm}, we stated that we expect $\mathrm{Bal}_n\asymp n$. Here we justify this claim.

The following lemma determines the scale of $\mathrm{Bal}_n$. For background, we recall that a \emph{Hadamard matrix} in dimension $n$ is a matrix $H_n\in\{-1,1\}^{n\times n}$ such that its columns are pairwise orthogonal. 
\begin{lemma}\label{lem:hadamard}
    Let $n$ be a positive integer that is a multiple of $4$. It holds that
    \[\mathrm{Bal}_n \leq n-1,\]
    with equality if and only if there exists a Hadamard matrix of dimension $n$.
\end{lemma}
\noindent A Hadamard matrix $H_n$ is conjectured to exist for any $n$ that is a multiple of $4$ (the \emph{Hadamard conjecture}), and is known to exist for certain special cases, e.g., any $n$ that is a power of $2$.  (See \citet{horadam2012hadamard} for additional background.) In other words, since we are assuming throughout this example that $n$ is a multiple of $4$ we expect that $\mathrm{Bal}_n\asymp n$.

\begin{proof}[Proof of Lemma~\ref{lem:hadamard}]
Let $\sigma_1,\dots,\sigma_k$ be a collection of permutations where $\sigma_i\circ\sigma_j^{-1}$ is balanced for each $i\neq j$. For each $i=1,\dots,k$, define a vector $v_i\in\{\pm1\}^n$ as
\[(v_i)_j = \begin{cases} +1, & A_{\sigma_i(j)}=1,\\ -1, & A_{\sigma_i(j)}=0.\end{cases}\]
Then, since $\sigma_i\circ\sigma_j^{-1}\in\cS_{n,\mathrm{bal}}$ for each $i\neq j\in[k]$, we must have $v_i \perp v_j$, since
\begin{align*}
    v_i^\top v_j
    &=\sum_{\ell=1}^n \left(\One{A_{\sigma_i(\ell)}=A_{\sigma_j(\ell)}} - \One{A_{\sigma_i(\ell)}\neq A_{\sigma_j(\ell)}}\right) \\
    &=\sum_{\ell=1}^n \left(\One{A_{\sigma_i\circ\sigma_j^{-1}(\ell)}=A_\ell} - \One{A_{\sigma_i\circ\sigma_j^{-1}(\ell)}\neq A_{\ell}}\right) \\
    &= n/2 -n/2 =0,
\end{align*}
where the last step holds since $\sigma_i\circ\sigma_j^{-1}\in\cS_{n,\mathrm{bal}}$, and the next-to-last step holds by replacing $\ell$ with $\sigma_j^{-1}(\ell)$ in the summation.
Moreover, $v_i\perp\mathbf{1}_n$, since $A$ contains equal numbers of $1$'s and $0$'s. Therefore, the vectors $\mathbf{1}_n,v_1,\dots,v_k\in\{\pm 1\}^n$ are mutually orthogonal, meaning that we must have $k+1\leq n$. Moreover, if $k=n-1$, then $H_n = (\mathbf{1}_n \, v_1 \, \dots \, v_{n-1})\in\{\pm 1\}^{n\times n}$ is a Hadamard matrix.

For the converse, suppose there exists a Hadamard matrix $H_n\in\{\pm1\}^{n\times n}$. We write $v_0,v_1,\dots,v_{n-1}\in\{\pm1\}^n$ to denote the columns of this matrix. Define $\tilde{H}_n = \mathrm{diag}(v_0)\cdot H_n$, which is also a Hadamard matrix, and now has first column $\mathbf{1}_n$. Let the subsequent columns of $\tilde{H}_n$ be $\tilde{v}_1,\dots,\tilde{v}_{n-1}$. For each $i=1,\dots,n-1$, let $\sigma_i$ be any permutation such that  the vector $(A_{\sigma_i(1)},\dots,A_{\sigma_i(n)})$ agrees with the vector $\tilde{v}_i$ (i.e., $A_{\sigma_i(j)}=1$ or $=0$, corresponds to $(\tilde{v}_i)_j=+1$ or $=-1$); note that $\tilde{v}_i$ must contain exactly $n/2$ entries of each sign, since $\tilde{v}_i\perp\mathbf{1}_n$, and therefore such a permutation must exist. Then by construction, the collection $\sigma_1,\dots,\sigma_{n-1}$ satisfies that $\sigma_i\circ\sigma_j^{-1}\in\cS_{n,\mathrm{bal}}$ for each $i\neq j$, since $\tilde{v}_i\perp\tilde{v}_j$.
\end{proof}

\end{document}